\documentclass[11pt]{article}
\usepackage[utf8]{inputenc}
\usepackage{float} 
\usepackage{amsmath}
\usepackage{amssymb}
\usepackage{algorithm}
\usepackage{booktabs}
\usepackage[noend]{algorithmic}
\usepackage{array}
\usepackage{caption}
\usepackage{subcaption}
\usepackage{bbm}
\usepackage{makeidx}
\usepackage[linktocpage=true]{hyperref}
\usepackage[capitalise]{cleveref}
\usepackage{amsthm}
\usepackage{color}
\usepackage{mathtools}
\usepackage{graphicx} 
\usepackage[LGR,T1]{fontenc}
\usepackage{siunitx}
\usepackage{mathrsfs}
\usepackage{scalerel}
\usepackage{forest}
\usepackage{lipsum}
\usepackage{xcolor}
\usepackage[title]{appendix}
\hypersetup{
    colorlinks,
    linkcolor={blue!80!black},
    citecolor={green!50!black},
}
\usepackage{accents}
\usepackage{apptools}
\usepackage{tikz}
\usetikzlibrary{arrows, positioning, automata}
\usepackage{comment}
\usepackage[shortlabels]{enumitem}

\usepackage[mathscr]{euscript}
\DeclareSymbolFont{rsfs}{U}{rsfs}{m}{n}
\DeclareSymbolFontAlphabet{\mathscrsfs}{rsfs}

\AtAppendix{\counterwithin{lemma}{section}}

\usepackage{mathtools}

\theoremstyle{plain}

\newtheorem{theorem}{Theorem}

\newtheorem{lemma}{Lemma}[section]
\newtheorem{example}{Example}[section]

\newtheorem{remark}{Remark}[section]
\let\oldremark\remark
\renewcommand{\remark}{\oldremark\normalfont}

\allowdisplaybreaks
\usepackage[top=1in,bottom=1in,left=1in,right=1in]{geometry}
\usepackage[utf8]{inputenc}

\def\bvarphi{\boldsymbol{\varphi}}

\def\TV{\mathsf{TV}}

\def\<{\langle}
\def\>{\rangle}
\def\dd{\mathrm{d}}

\def\top{\intercal}

\def\op{\mbox{\rm \tiny op}}

\def\cS{\mathcal{S}}
\def\cA{\mathcal{A}}

\def\Lamp[#1]{\boldsymbol{\Lambda}_{\mathrm{AMP}}^{(#1)}}
\def\lalg[#1]{\Lambda_{\mathrm{alg}, #1}}

\def\de{{\rm d}}

\def\EE{\mathbb{E}}

\def\NN{\mathbb{N}}

\def\RR{\mathbb{R}}

\def\bA{\mathbf{A}}

\def\bM{\mathbf{M}}

\def\bX{\mathbf{X}}

\def\ba{\boldsymbol{a}}

\def\bg{\boldsymbol{g}}
\def\bh{\boldsymbol{h}}

\def\bv{\boldsymbol{v}}

\def\bx{\boldsymbol{x}}
\def\by{\boldsymbol{y}}

\def\bA{\boldsymbol{A}}

\def\bM{\boldsymbol{M}}

\def\bX{\boldsymbol{X}}

\def\normal{{\mathsf{N}}}

\def\btheta{\boldsymbol{\theta}}

\def\bSigma{\boldsymbol{\Sigma}}

\def\bgamma{\boldsymbol{\gamma}}
\def\bOmega{\boldsymbol{\Omega}}

\def\bepsilon{\boldsymbol{\varepsilon}}
\def\beps{\boldsymbol{\varepsilon}}

\def\brho{\boldsymbol{\rho}}

\def\id{{\boldsymbol I}}

\def\sT{{\sf T}}

\renewcommand{\P}{\mathbb{P}}
\newcommand{\E}{\mathbb{E}}
\newcommand{\R}{\mathbb{R}}

\newcommand{\eps}{\varepsilon}
\newcommand{\Var}{\operatorname{Var}}

\newcommand{\Cov}{\operatorname{Cov}}

\newcommand{\diag}{\operatorname{diag}}

\newcommand{\Unif}{\operatorname{Unif}}

\newcommand{\Laplace}{\operatorname{Laplace}}

\newcommand{\RN}[1]{%
  \textup{\uppercase\expandafter{\romannumeral#1}}%
}

\newcommand\iidsim{\sim_{i.i.d.}}

\newcommand{\RNum}[1]{\uppercase\expandafter{\romannumeral #1\relax}}



\makeatletter
\newcommand*{\rom}[1]{\expandafter\@slowromancap\romannumeral #1@}
\makeatother

\title{Provably Efficient Posterior Sampling for Sparse Linear Regression via
Measure Decomposition}
\author{
	Andrea Montanari\thanks{Department of Statistics and Department of Mathematics, Stanford University} 
	\and 
	Yuchen Wu\thanks{Department of Statistics and Data Science, Wharton School, University of Pennsylvania}
}

\date{\today}
\begin{document}

\maketitle

\begin{abstract}
We consider the problem of sampling from the posterior distribution of a 
$d$-dimensional coefficient vector $\btheta$, given linear observations 
$\by = \bX\btheta+\beps$. 
In general, such posteriors are multimodal, and therefore challenging to sample from.
This observation has prompted the exploration of various heuristics that aim
at approximating the posterior distribution.

In this paper, we study a different approach based on decomposing 
the posterior distribution into a log-concave mixture of simple product measures. 
This decomposition allows us to reduce sampling from a multimodal distribution of
interest to sampling from a log-concave one, which is tractable and has been
investigated in detail.
We prove that, under mild conditions on the prior, for random designs,
such measure decomposition 
is generally feasible when the number of samples per parameter $n/d$
exceeds a constant threshold. 
We thus obtain a provably efficient (polynomial time) sampling algorithm 
in a regime where this was previously not known. 
Numerical simulations confirm that the algorithm is practical, and reveal that it has
attractive statistical properties compared to state-of-the-art methods.
\end{abstract}

\tableofcontents

\section{Introduction}
\label{sec:introduction}

We consider a standard linear model 
\begin{align}
\label{eq:model-LR}
	\by = \bX \btheta + \beps,
\end{align}
where $\bX = [\bx_1|\cdots|\bx_n]^{\sT} \in \RR^{n \times d}$ denotes the design matrix, 
$\by = (y_1, \cdots, y_n)^{\top} \in \RR^n$ is the response vector,
 $\beps = (\eps_1, \cdots, \eps_n)^{\top} \sim \normal(\mathbf{0}_n, \sigma_d^2\id_n)$ is the
  noise vector, and $\btheta \in \RR^d$ denotes the hidden coefficients.
We investigate model \eqref{eq:model-LR} under a high-dimensional and sparse setup where 
the number of model parameters $d$ is comparable to or even larger than the sample size $n$, 
and a substantial proportion of the entries of the coefficient vector $\btheta$ are zero. 
Observing the pair $(\by, \bX)$, our objective is to conduct inference on $\btheta$. 
This high-dimensional regression problem has been widely studied both within
the Bayesian and the frequentist communities
   \cite{mitchell1988bayesian,george1993variable,tibshirani1996regression,efron2004least,
   mj2009sharp,narisetty2014bayesian,rovckova2018spike,bertsimas2020sparse}.

In a Bayesian approach, we endow the coefficient vector $\btheta$ with a prior distribution 
 $\pi$ over $\RR^d$, then the posterior distribution upon observing $(\by, \bX)$ takes the form 
\begin{align}
	\pi(\de \btheta \mid \by, \bX) = \frac{1}{Z_0(\by, \bX)} \exp\left( - \frac{1}{2\sigma_d^2} \btheta^{\top }\bX^{\top} \bX \btheta + \frac{1}{\sigma_d^2}\btheta^{\top}\bX^{\top} \by 
	\right) \, \pi(\de \btheta)\, ,\label{eq:FirstPosterior}
\end{align}
where $Z_0(\by, \bX)$ is a normalizing constant that is a function of $(\by, \bX)$.
In order to establish uncertainty quantification for $\btheta$, Bayesian methods require an 
algorithm that efficiently draws samples from the posterior distribution \eqref{eq:FirstPosterior}.
It is worth noting that a separate line of work 
focuses instead on computing the posterior mode \cite{rovckova2014emvs,rovckova2018spike}.
This is also known as  maximum a posteriori estimation. 
However, mode detection does not provide --in general-- a method for uncertainty quantification, 
and posterior sampling is generally regarded as a more challenging task.

Bayesian regression has demonstrated state-of-the-art performance across many application domains  \cite{tipping2001sparse,guan2011bayesian,ikehata2014photometric,wang2019robust}.
It also enjoys broad popularity 
for solving linear inverse problems in a variety of scientific fields,
ranging from geology to medical imaging \cite{stuart2010inverse,nickl2023bayesian}.
Among the various options to perform sparse regression within the Bayesian framework, 
methods based on the \emph{spike-and-slab prior} are the default choice \cite{bai2021spike}.
The spike-and-slab prior was first proposed in \cite{george1997approaches}, and has since served as an important building block in Bayesian statistics. 
For readers' convenience, we present a brief overview of the spike-and-slab prior in
 \cref{sec:spike-and-slab}, 
and discuss several prominent sampling algorithms associated with it in  \cref{sec:prior-arts}. 

Despite the continued progress in developing posterior sampling algorithms,
the accompanying theoretical guarantees are less satisfactory.  
\begin{figure}
	\centering
	\includegraphics[width=\textwidth]{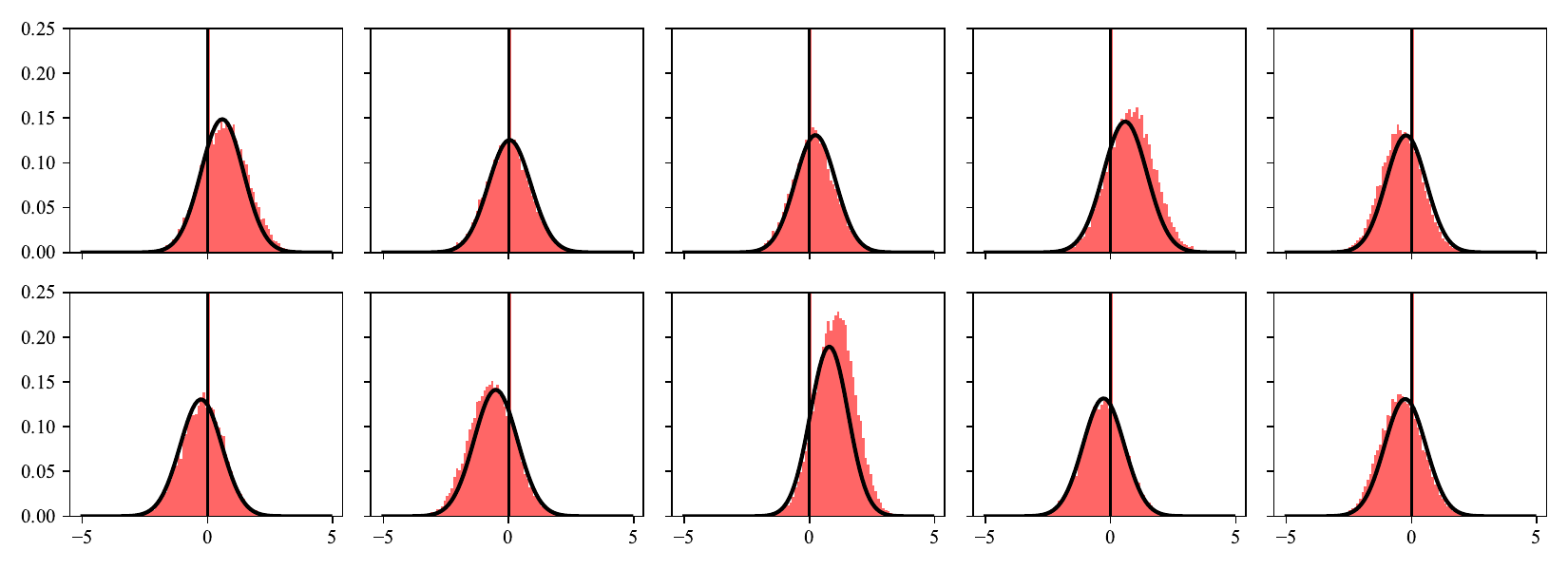}
	\caption{True and approximated posterior distributions produced by the proposed two-stage sampling algorithm. In this figure, we set $q = 0.3$, $\mu = \normal(0, 1)$ and $\sigma_d = 1$.
	We take $n = 20$ and $d = 10$. We generate the design matrix $\bX$ randomly via $X_{ij} \iidsim \normal(0, 1 / 4d)$. 
	We sample $\by = \bX \btheta + \beps$, where $\theta_i \iidsim q\, \delta_0 + (1 - q) \mu$ and $\beps \sim \normal(\mathbf{0}_n, \sigma_d^2\id_n)$. 
	We fix $(\bX, \by)$ after they are generated, and aim to sample from the associated posterior distribution. 
	In the above figure, the red bins represent the empirical sample distribution, and the
	 black line indicates the true posterior.
	 Different subplots present the empirical sample distributions and the true posteriors for $d$ different coordinates. 
	}
	\label{fig:shape}
\end{figure}
A line of research investigates posterior contraction properties
in the sparse high-dimensional regime  under frequentist assumptions on the data distributions
\cite{castillo2015bayesian,rovckova2018spike,shin2022neuronized}.
While these works support the use of Bayesian regression
methods, they do not provide algorithms to sample from the target posterior. 

A separate line of research focuses on designing and analyzing Markov Chain Monte Carlo (MCMC)
 algorithms for posterior sampling \cite{belloni2009computational,richardson2010bayesian,
 schreck2015shrinkage,yang2016computational}.
However, theoretical analysis of MCMC mixing time is notoriously challenging. 
Existing theoretical guarantees only apply to regimes in which statistical 
uncertainty is small and the posterior has a simple structure.
For instance,  \cite{belloni2009computational} proves mixing when the dimension
$d$ grows moderately as compared to the sample size $n$, and the posterior is approximately normal.
The high-dimensional case is covered in \cite{yang2016computational},
which requires however irrepresentability-type conditions on the design matrix.
Under these conditions, the posterior concentrates around vectors with a fixed set of non-zeros,
and (because of the structure of the prior) is approximately normal.
In general, constructing Markov chains that enjoy fast mixing properties is elusive 
even under simple statistical models, let alone having quantitative control of the mixing time.

Variational inference approaches provide another useful toolkit for Bayesian inference
\cite{jordan1999introduction,wainwright2008graphical,blei2017variational}. These methods replace the actual posterior by its closest approximation
within a specific parametric family, thus effectively replacing sampling with optimization.
Normally, the approximating family consists of product measures,
an ansatz known as `naive mean field.' While positive guarantees have been established for
naive mean field in certain settings \cite{ray2022variational,mukherjee2022variational},
in general the variational inference approach incurs uncontrolled approximation errors. 
For instance, \cite{ghorbani2019instability} proves that ---in a simple high-dimensional problem--- the posterior mean 
computed by naive mean field can be arbitrarily wrong, even when the prior takes a simple product form.

In this paper, we propose a new class of sampling algorithms for Bayesian linear regression,
which are constructed by decomposing the target posterior into a mixture of product measures.
We prove that, for a broad class of priors, and for isotropic random designs,
the mixture distribution can be sampled efficiently, provided that the number of samples
per parameter $n/d$ is larger than a constant threshold. 
This theoretical guarantee covers a regime in which (under the posterior) the support of
the coefficient vector is non-deterministic, hence opening the way to uncertainty quantification
for the support.
As a consequence, we obtain an efficient sampling algorithm in a regime in which no comparable 
results exist.

Our proposal is directly inspired by  recent advances in probability theory  that develop new techniques to bound the log-Sobolev constant of spin models \cite{bauerschmidt2019very,eldan2022spectral}.  
Specifically, the approach from \cite{bauerschmidt2019very,eldan2022spectral} enables us to analyze various properties of non-log-concave measures by decomposing them into mixtures of simpler ones. 
Our goal is to develop and study the algorithmic versions of these ideas. 

As shown by numerical studies in \cref{sec:experiments}, 
our approach is simple, effective, and compatible with any black-box sampling algorithm that is able to sample from log-concave distributions.  
As a preview of our results,  
Figure \ref{fig:shape} presents the empirical distributions associated with individual coefficients produced by our sampling algorithm in a small-scale example,
where the true posterior can be computed exactly.
Comparing the empirical distributions with the true posteriors, we observe a close match.  

The remainder of the paper is structured as follows.
In \cref{sec:preliminaries}, we formulate the sampling problem and discuss the spike-and-slab prior 
along with its continuous relaxations. 
In \cref{sec:sampling-alg} we describe the sampling algorithm based  on
measure decomposition and state the theoretical guarantee for our proposal.
Finally, we present numerical experiments that support our findings in \cref{sec:experiments}.

\subsection{Notations}

For $n \in \NN_+$, we denote by $[n]$ the set that contains all positive integers from 1 to $n$. For $a, b \in \RR$, we denote by $a \vee b$ the maximum of $a$ and $b$.
For two distributions $\mu_1$, $\mu_2$ and a real number $q \in [0, 1]$, we use $q \mu_1 + (1 - q) \mu_2$ to denote the mixture of these two distributions with mixing probability $q$.
For a matrix $\bX$, we denote by $\|\bX\|_{\op}$ its operator norm, 
$\lambda_{\min}(\bX)$ its minimum eigenvalue, and $\lambda_{\max}(\bX)$ its
 maximum eigenvalue.  
We use $\|\bv\|$ to denote the Euclidean norm of a vector $\bv$, and use
$\TV(\mu, \nu)$ to denote the total variation distance between 
measures $\mu$ and $\nu$. 
For a random variable $X$, we use $\|X\|_{\psi_2}$ to denote its sub-Gaussian norm. See \cite[Section 2.5.2]{vershynin2018high} for a formal definition of sub-Gaussian norm.

\section{Preliminaries}
\label{sec:preliminaries}

The spike-and-slab prior has appealing statistical properties but simultaneously 
poses significant challenges to standard sampling algorithms. 
In this section, we provide background on the spike-and-slab prior. 
For clarity of exposition, we will focus on a simplified version of this prior, and we will discuss generalizations later. 
%

\subsection{The spike-and-slab prior and its continuous relaxation}
\label{sec:spike-and-slab}

Numerous priors have been proposed and analyzed in the literature.
These  typically take the form of a mixture of product distributions with few latent variables. 
Namely, the prior admits the following decomposition:
\begin{align}
\label{eq:prior-decomposition}
	\pi(\de\btheta) = \int \pi_0^{\otimes d}(\de\btheta|\rho)\, 
\pi_{\rho}(\de \rho). 
\end{align}
In the above display, $\rho$ represents a vector of latent variables, the size of which is typically small and independent of the problem scale. 
Given $\rho$, $\pi_0(\dd \theta \mid \rho)$ denotes a distribution over $\RR$, and we use $\pi_0^{\otimes d}(\de\btheta|\rho)$ to represent a product distribution over $\RR^d$ with coordinate-wise marginal distribution $\pi_0(\dd \theta \mid \rho)$. 
We list below several prominent examples of prior distributions that admit 
the representation  \eqref{eq:prior-decomposition}. 
In particular, we feature the {spike-and-slab priors} and their continuous relaxations.

\begin{example}[Spike-and-slab priors]
\label{example:spike-and-slab}
	The spike-and-slab prior was first proposed in \cite{george1993variable}, and usually takes the following form: 
	\begin{align}
	\label{eq:spike-and-slab}
	\begin{split}
		& \pi_0^{\otimes d}(\dd \btheta \mid \bgamma, \sigma^2) = \prod_{j = 1}^d \left[ (1 - \gamma_j) \delta_0 + \gamma_j \mu(\dd \theta_j \mid \sigma^2) \right], \\
		& \pi (\dd \bgamma \mid q) = \prod_{j = 1}^d q^{\gamma_j}(1 - q)^{1 - \gamma_j}, \qquad q \sim \pi_q(\dd q), \qquad \sigma^2 \sim \pi_{\sigma^2}(\dd \sigma^2). 
	\end{split}
	\end{align}
	In the above display, $\delta_0$ stands for a point mass distribution at zero, 
	$\mu(\dd \theta \mid \sigma^2)$ is a diffuse density that scales with $\sigma^2$, 
	$\bgamma \in \{0, 1\}^d$ is a binary vector, and $q \in (0, 1)$ is the mixing probability. 
	We further assume that $q$ and $\sigma^2$ follow prior distributions $\pi_q$ and $\pi_{\sigma^2}$, 
	respectively. 
	
We note that the above example fits in the general setting of Eq.~\eqref{eq:prior-decomposition}
	after we marginalize over the selection variables $\bgamma$.
In this case, the latent variables are $\rho = (q, \sigma^2)$, and $\pi_0(\dd \theta \mid \rho)$
	 is the probability distribution of $(1 - q) \delta_0 + q \mu(\dd \theta \mid \sigma^2)$
	  marginalizing over $q \sim \pi_{q}$ and $\sigma^2 \sim \pi_{\sigma^2}$.  
\end{example}

The point-mass spike-and-slab prior given in \cref{eq:spike-and-slab} is considered the theoretical gold standard for Bayesian variable selection \cite{johnstone2004needles,ishwaran2011consistency,castillo2012needles,polson2019bayesian}.
However, sampling from the corresponding posterior  can be computationally prohibitive due to the combinatorial nature of $\bgamma$. 
As an alternative, researchers have resorted to continuous relaxations of \eqref{eq:spike-and-slab}, which replaces $\delta_0$ 
with a density that is peaked at zero. 
\begin{example}[Continuous relaxations of the spike-and-slab priors]
\label{example:continuous-relaxation}
	We let:
	\begin{align}
	\label{eq:spike-and-slab-relax}
	\begin{split}
		& \pi_0^{\otimes d}(\dd \btheta \mid \bgamma, \sigma^2) = \prod_{j = 1}^d \left[ (1 - \gamma_j) \mu_0(\dd \theta_j \mid \sigma^2) + \gamma_j \mu_1(\dd \theta_j \mid \sigma^2) \right], \\
		& \pi (\dd \bgamma \mid q) = \prod_{j = 1}^d q^{\gamma_j}(1 - q)^{1 - \gamma_j}, \qquad q \sim \pi_q(\dd q), \qquad \sigma^2 \sim \pi_{\sigma^2}(\dd \sigma^2). 
	\end{split}
	\end{align}  
	Again, the above prior fits in the setting of Eq.~\eqref{eq:prior-decomposition}.
	As mentioned, in \eqref{eq:spike-and-slab-relax}, the point-mass distribution $\delta_0$ is
	 replaced by a continuous distribution $\mu_0$ that concentrates around 0. 
	 
	 	Among others, \cite{george1993variable} proposed  to use a Gaussian mixture prior 
 $\mu_0=\normal(0,\sigma_0^2)$ and $\mu_1= \normal(0,\sigma_1^2)$ with $\sigma_0\ll \sigma_1$;
More recently, \cite{rovckova2018spike} proposed the spike-and-slab LASSO prior  with
 $\mu_0={\sf Laplace}(\lambda_0)$ and $\mu_1= {\sf Laplace}(\lambda_1)$
 with $\lambda_0\gg \lambda_1$\footnote{Here, we assume ${\sf Laplace}(\lambda)$ has density 
 $\frac{\lambda}{2} \exp(-\lambda |x|)$ for $x \in \RR$. }, and established 
 minimax optimality for this proposal.
Note that the spike-and-slab LASSO prior simply reduces to the Lasso prior when equalizing 
$\lambda_1$ and $\lambda_0$. 
The examples mentioned here are within the broader family of global-local shrinkage priors. 
We refer interested readers to Table 2 in \cite{bhadra2019lasso} for a survey of these priors.
\end{example} 

Despite benefiting from the continuous relaxation, posterior sampling with spike-and-slab priors 
remains challenging and there is no algorithm with sampling guarantees in the noisy 
high-dimensional regime tackled by our work. 
We refer the readers to \cref{sec:prior-arts} for a discussion on several previous sampling algorithms in this direction.

\subsection{Problem formulation and challenges}

For the sake of simplicity, in this paper we restrict to prior distributions that 
take product forms, i.e., we assume $\rho \equiv 1$ in \cref{eq:prior-decomposition}. 
This can be equivalently viewed as fixing a value of the latent variable $\rho$ and sampling from 
the posterior distribution conditioning on $(\by, \bX, \rho)$ instead of $(\by, \bX)$. 
In order to sample from a hierarchical model with latent variables, we may resort 
to several strategies. A popular one is to use
 a Gibbs sampler that alternates between sampling from $\pi(\dd \btheta \mid \by, \bX, \rho)$ 
and $\pi(\dd \rho \mid \by, \bX, \btheta)$. 
We expect sampling from $\pi(\dd \rho \mid \by, \bX, \btheta)$ to be tractable, 
since $\rho$ is typically a low-dimensional vector.
An alternative would be to sweep over a grid of values of $\rho$ and use our algorithm to
estimate the posterior weights $\pi(\dd \rho \mid \by, \bX)$. 

We leave the question of sampling from the low-dimensional latent vector $\rho$
for future work, and instead focus on what we consider the crux of the problem,
namely, sampling from the posterior distribution associated with 
 a product prior. Equivalently, we wish to sample from $\pi(\dd \btheta \mid \by, \bX, \rho)$ 
 if we view $\pi(\dd \btheta \mid \rho)$ as the prior. 
To simplify things, throughout this work we drop $\rho$ since it is understood to be fixed.

We will focus on the point-mass spike-and-slab prior 
 defined in Example \ref{example:spike-and-slab}, but generalizations to other priors are immediate.
To be precise, we assume $\btheta$ has a product prior with marginal distribution 
\begin{align}
\label{eq:mixture-prior}
\pi_0(\de\theta) = (1-q)\, \delta_0 + q\, \mu(\de\theta). 
\end{align}
As we have mentioned, for prior distributions that admit form \eqref{eq:mixture-prior}, the 
associated posterior \eqref{eq:FirstPosterior}
is in general not log-concave, and standard sampling algorithms
(e.g., Langevin dynamics, Hamiltonian Monte Carlo, and their variants \cite{luengo2020survey})
come with no theoretical guarantees. 
Gibbs sampling can be attempted, but standard analysis methods
(e.g., those based on checking the Dobrushin condition \cite{dobruschin1968description}) only allow to establish fast mixing under very restrictive assumptions.

These limitations are compounded by the remark that, in general, 
sampling from the above Bayes posterior is NP-hard.
For instance, in the case with $\mu= \Unif([-M,M])$,
sampling from the target posterior is at least as hard as minimum cardinality regression (i.e., $\ell_0$-norm regularization), which is NP-hard  by
\cite{natarajan1995sparse}.

To summarize, the goal of this paper is to design an efficient sampling algorithm for the posterior distribution \eqref{eq:FirstPosterior}, within the framework of linear model \eqref{eq:model-LR} that has prior \eqref{eq:mixture-prior}. 




\section{Sampling based on measure decomposition}
\label{sec:sampling-alg}

We describe in this section our sampling algorithm. 
At a high level, we decompose the target posterior into a mixture of  product measures. 
To achieve this, we introduce an intermediate variable $\bvarphi \in \RR^d$, such that 
when conditioned on $(\bvarphi, \by, \bX)$, the variable $\btheta$ has a product 
conditional distribution.
Furthermore, we prove that under certain conditions, $\bvarphi$ has a
 log-concave density, hence is amenable to efficient sampling. 
 \cref{sec:log-concave} reviews algorithms that efficiently
 sample from log-concave distributions. 

\subsection{Measure decomposition}
\label{sec:measure-decomposition}

Let $\gamma$ be a positive constant, such that the matrix $\bA := \gamma \id_d - \sigma_d^{-2} \bX^{\top} \bX$ is strictly positive-semidefinite. 
Namely,  it suffices to take $\gamma > \sigma_d^{-2} \|\bX\|_{\op}^2$. 
The target posterior \eqref{eq:FirstPosterior} then takes the following form (recall that we assume $\pi(\dd \btheta) = \pi_0^{\otimes d}(\dd \btheta)$): 
\begin{align}
\label{eq:theta-posterior}
	\pi(\de \btheta \mid \by, \bX) = \frac{1}{Z_0(\by, \bX)} \exp\left( \langle \btheta, \bh \rangle + \frac{1}{2} \big\langle \btheta, \bA \btheta \big\rangle - \frac{\gamma}{2}\|\btheta\|^2 \right) \pi_0^{\otimes n} (\de \btheta), 
\end{align}
where $\bh = \sigma_d^{-2}\bX^{\top} \by \in \RR^d$.
 
\cref{lemma:marginal}  shows that density \eqref{eq:theta-posterior} corresponds to the 
marginal distribution for the first $d$ coordinates of a joint distribution over $(\btheta, \bvarphi) \in \RR^d \times \RR^d$.  
\begin{lemma}[Measure decomposition]
\label{lemma:marginal}
	Assume $\gamma > \sigma_d^{-2} \|\bX\|_{\op}^2$. Then distribution \eqref{eq:theta-posterior} is the marginal distribution for the first $d$ coordinates of the following joint distribution:
\begin{align}\label{eq:pi}
	\pi(\de \btheta, \de \bvarphi \mid \by, \bX) \propto \exp \left( \langle \bh + \bvarphi, \btheta \rangle - \frac{1}{2} \big\langle \bvarphi, \bA^{-1} \bvarphi \big \rangle - \frac{\gamma}{2} \|\btheta\|^2  \right)\pi_0^{\otimes n} (\de \btheta)\, \de \bvarphi. 
\end{align}
Here, $\de \bvarphi$ denotes the Lebesgue measure over $\RR^d$. 
\end{lemma}
\begin{proof}[Proof of \cref{lemma:marginal}]
	
	Integrating the quantity on the right hand side of \cref{eq:pi} over $\bvarphi$, we see that 
	\begin{align*}
		& \int_{\RR^d} \exp \left( \langle \bh + \bvarphi, \btheta \rangle - \frac{1}{2} \big\langle \bvarphi, \bA^{-1} \bvarphi \big \rangle - \frac{\gamma}{2} \|\btheta\|^2  \right)\pi_0^{\otimes n} (\de \btheta)\, \de \bvarphi \\
		= & \, C_{\bA} \exp \left( \langle \bh, \btheta \rangle + \frac{1}{2} \langle \btheta, \bA \btheta \rangle - \frac{\gamma}{2} \|\btheta\|^2 \right) \pi_0^{\otimes n} (\de \btheta), 
	\end{align*}
	where $C_{\bA}$ is a constant that depends only on $\bA$. The second line above coincides with the right hand side of \cref{eq:theta-posterior} up to a normalizing constant, thus concluding the proof. 
	
\end{proof}
Equation \eqref{eq:pi} characterizes the joint distribution of $(\btheta, \bvarphi)$. 
Using the density function written there, we conclude that the marginal distribution of $\bvarphi$ takes the form
\begin{align*}
	\pi(\de \bvarphi \mid \by, \bX) \propto e^{-H(\bvarphi)} \de \bvarphi, 
\end{align*}
where
\begin{align*}
	& H(\bvarphi) := \frac{1}{2} \langle \bvarphi, \bA^{-1} \bvarphi  \rangle + \sum_{i = 1}^d V_{\gamma} (h_i + \varphi_i), \\
	& V_\gamma (x) := -\log\left\{ \int_{\RR} e^{x \theta - \frac{\gamma}{2} \theta^2}  \pi_0(\de \theta) \right\}. 
\end{align*}
As we will see in \cref{lemma:V-gamma-general}, the second derivative of $V_{\gamma}$ is non-positive.

The conditional distribution of $\btheta$ conditioning on $(\bvarphi, \by, \bX)$ then admits a product form: 
\begin{align}
\label{eq:conditional-product}
	\pi(\de \btheta \mid \bvarphi, \by, \bX) \propto \prod_{i = 1}^d e^{(h_i + \varphi_i)\theta_i - \frac{\gamma}{2}\theta_i^2} \pi_0(\de \theta_i). 
\end{align}

\begin{remark}
The decomposition in Lemma \ref{lemma:marginal} has been used for a long time in statistical
physics \cite{baker1962certain,hubbard1972critical} to study spin models which are  
probability measures of the form \eqref{eq:theta-posterior}.
To the best of our knowledge, \cite{bauerschmidt2019very} first noticed that the shift term $\gamma\id_d$
can be exploited to simplify the structure of the marginal distribution of $\bvarphi$.
\end{remark}

\subsection{Two-stage sampling algorithm}
\label{sec:two-stage}

The measure decomposition presented in \cref{sec:measure-decomposition} 
suggests the following two-stage algorithm: 
\begin{enumerate}
	\item First, we sample $\bvarphi \sim \pi(\de \bvarphi \mid \by, \bX)$. We denote by $\cA_1$ the sampling algorithm for this. 
	\item Given $\bvarphi$, we sample $\btheta$ from the corresponding conditional distribution $\btheta \sim \pi(\de \btheta \mid \bvarphi, \by, \bX)$. The algorithm used to sample in this step is denoted by $\cA_2$.  
\end{enumerate}
We note that Step 2 of the above procedure in general can be implemented efficiently, as the conditional distribution takes a product form and each component is simply a tilted version of the prior distribution $\pi_0$. 
As for step 1, a sufficient condition under which this can be efficiently implemented is that
$H(\bvarphi)$ is strongly convex.
In this case, we can leverage the rich and rapidly growing literature on sampling
log-concave distributions to construct $\cA_2$, 
see \cref{sec:log-concave} for background. 
In this case, we say that the sampling problem is \emph{feasible}. 
Inspecting the Hessian of $H(\bvarphi)$, we conclude that sampling is feasible if and only if there exists $\gamma > \sigma_d^{-2} \|\bX\|_{\op}^2$, such that 
\begin{align}
\label{eq:def-feasible}
	\frac{1}{\gamma - \lambda_{\min}(\sigma_d^{-2} \bX^{\top} \bX)} + \inf_{x \in \RR} V_{\gamma}''(x)  > 0. 
\end{align}
We note that condition \eqref{eq:def-feasible} is straightforward to verify given an estimate of the noise level, and is independent of the response $\by$. 
We next present a sufficient condition for the convexity of $H(\bvarphi)$. 
To this end, we establish the following lemma.  
\begin{lemma}\label{lemma:V-gamma-general}
Recall that $\mu$ is the diffuse density and $q$ is the mixing probability, both given in \cref{eq:mixture-prior}. 
Assume that $\mu$ has a log-concave density $f_{\mu}$ that is symmetric about the origin.
Further assume that there exist $c_1$, $c_2 \in \RR_{> 0}$ and $k \in \NN_+$ that depend only on $\mu$, such that $f_{\mu} (x) \geq c_1 e^{-c_2 x^{2k}}$ for all $x \in \RR$. Then, there exists a constant $C_0 > 0$ that depends only on $(q, \mu)$, such that 
	\begin{align}
	\label{eq:supV''}
		\inf_{x \in \RR}V_{\gamma}''(x) \geq -C_0(\gamma^{-1} + \gamma^{-2}) \cdot (1 + \log (\gamma + 1))^{\frac{2k - 1}{k}}.
	\end{align}
	In addition, $V_{\gamma}''(x) \leq 0$ for all $x \in \RR$. 
\end{lemma}

\begin{remark}
\label{rmk:mean-zero-and-log-concave}
	The assumption that $\mu$ is mean-zero and log-concave is a common characteristic of many widely used distributions. 
	To name a few, see \cite{mitchell1988bayesian,rovckova2018bayesian,polson2019bayesian}. 
\end{remark}

\begin{proof}[Proof of \cref{lemma:V-gamma-general}]
	We defer the proof of  \cref{lemma:V-gamma-general} to Appendix \ref{sec:proof-lemma:V-gamma-general}.
\end{proof}

\cref{lemma:V-gamma-general} lower bounds the second derivative of $V_{\gamma}$. 
We can then leverage this lemma to lower bound the eigenvalues of $\nabla^2 H(\bvarphi)$. 
More precisely, under the conditions of \cref{lemma:V-gamma-general},
\begin{align*}
	\lambda_{\min} (\nabla^2 H(\bvarphi)) = & \lambda_{\min} \Big( \bA^{-1} + \diag(\{V_{\gamma}''(h_i + \varphi_i)\}_{i \in [d]}) \Big) \\
	\geq & \frac{1}{\gamma - \lambda_{\min} (\sigma_d^{-2} \bX^{\top} \bX)} - C_0(\gamma^{-1} + \gamma^{-2}) \cdot (1 + \log (\gamma + 1))^{\frac{2k - 1}{k}}.  
\end{align*}
As a consequence, we obtain that $H(\bvarphi)$ is strongly convex when 
\begin{align}
\label{eq:convex-condition}
	\frac{1}{\gamma - \lambda_{\min} (\sigma_d^{-2} \bX^{\top} \bX)} > C_0(\gamma^{-1} + \gamma^{-2}) \cdot (1 + \log (\gamma + 1))^{\frac{2k - 1}{k}}. 
\end{align}
\cref{eq:convex-condition} provides a sufficient condition that ensures the log-concavity of $\bvarphi$. 

In the following, we refer to the sampling problem as {feasible} if 
there exists $\gamma > \|\bX^{\top} \bX / \sigma_d^2\|_{\op}$, such that $\bvarphi$ under 
this choice of $\gamma$ has log-concave marginal density. 
When this happens, the proposed two-stage sampling algorithm has provable guarantees.
In the next section, we study  the feasible
 region for random design matrices.

To conclude this section, we demonstrate that if both $\cA_1$ and $\cA_2$ achieve high accuracy in their respective tasks, 
then combining them leads to a two-stage sampling algorithm with overall high accuracy. 

\begin{theorem}
\label{thm:combine}
Denote by $\hat\pi (\dd \bvarphi \mid \by, \bX)$ the output distribution of $\cA_1$ given $(\by, \bX)$, and denote by $\hat\pi (\dd \btheta \mid \bvarphi, \by, \bX)$ the output distribution of $\cA_2$ given $(\by, \bX, \bvarphi)$. 
	Assume that 
	\begin{align*}
		& \TV\left(\pi (\dd \bvarphi \mid \by, \bX), \hat\pi (\dd \bvarphi \mid \by, \bX)\right) \leq \eps_1, \\
		& \TV\left(\pi (\dd \btheta \mid \bvarphi, \by, \bX), \hat\pi (\dd \btheta \mid \bvarphi, \by, \bX) \right) \leq \eps_2, \qquad \forall \bvarphi \in \RR^d. 
	\end{align*}
	Let $\hat\pi(\dd \btheta \mid \by, \bX)$ be the output distribution of the proposed two-stage sampling algorithm. Then, it holds that 
	\begin{align*}
		\TV\left(\pi(\dd \btheta \mid \by, \bX), \hat\pi(\dd \btheta \mid \by, \bX) \right) \leq \eps_1 + \eps_2 + \eps_1 \eps_2. 
	\end{align*}
\end{theorem}
\begin{proof}[Proof of \cref{thm:combine}]
	By assumption, we can couple random variables sampled from the two distributions 
	$\pi(\dd \btheta \mid \by, \bX)$ and $\hat\pi(\dd \btheta \mid \by, \bX)$, such 
	that they are equal with probability at least $(1 - \eps_1)(1 - \eps_2)$. 
\end{proof}


\subsection{The case of random designs}
\label{sec:asymptotic-feasibility}


In this section, we discuss the feasibility of measure decomposition in the context of 
 random design matrices. 
Specifically, we examine two situations in which feasibility (in the sense of \cref{eq:def-feasible})
holds with high probability when $n/d$ is larger than a suitable constant.

\paragraph{Isotropic sub-Gaussian rows.} 

In the first situation, we assume that the rows of $\bX$ are independent, isotropic, 
and sub-Gaussian random vectors in $\RR^d$. 
We state our result below and defer the proof to Appendix \ref{sec:proof-thm:feasible}. 
Note that since the norm of $\btheta$ scales like $\|\btheta\|_2\asymp \sqrt{d}$,
the assumption $c_1 > \sigma_d^2 / d > c_2$ amounts to requiring that the
signal-to-noise ratio (SNR) is of order one.
\begin{theorem}
	\label{thm:feasible}
	Denote by $\bx_1, \bx_2, \cdots, \bx_n \in \RR^d$ the rows of $\bX$
	and assume the following:  
	$(1)$ The random vectors $\bx_1, \bx_2, \cdots, \bx_n$ are independent. 
	$(2)$ The rows are isotropic, in the sense that $\EE[\bx_i \bx_i^{\top}] = \id_d$ for 
	all $i \in [n]$. 
	$(3)$ The rows are sub-Gaussian random vectors, with a uniformly upper bounded sub-Gaussian constant:
	 $\max_{i \in [n]} \|\bx_i\|_{\psi_2} = K  < \infty$. 
	$(4)$ There exist numerical constants $c_1, c_2 > 0$, such that $c_1 > \sigma_d^2 / d > c_2$. 
	$(5)$ The conditions of \cref{lemma:V-gamma-general} hold.
	
	Under these assumptions, there exists a constant $C_1 > 0$ that depends only on $(q, \mu)$, such that if the following two conditions are satisfied:  
	\begin{align*}
	\frac{n}{d} \geq 4C_1 K^4 , \qquad \frac{\sqrt{d / n}}{ 2K^2 } \geq C_1 (d / n + d^2 / n^2) \cdot (1 + \log (n / d + 1))^{\frac{2k - 1}{k}}, 
	\end{align*}
	then with probability $1 - 2\exp(-d)$ the sampling problem is feasible. 
	Namely, there exists $\gamma > \sigma_d^{-2} \|\bX\|_{\op}^2$ such that \cref{eq:def-feasible} holds. 
\end{theorem}
\begin{proof}[Proof of \cref{thm:feasible}]
	We prove \cref{thm:feasible} in Appendix \ref{sec:proof-thm:feasible}. 
\end{proof}

As claimed above, \cref{thm:feasible} implies that the sampling problem is with high probability
 feasible given that $n / d$ is larger than a suitable constant.

\paragraph{Independent and identically distributed design.}
When $\bX$ contains i.i.d. entries, 
we can utilize tools from random matrix theory to precisely delineate the 
asymptotic feasible region.  
Throughout this example, we assume $X_{ij} \iidsim \mu_X$, where $\mu_X$ has mean zero, unit variance and finite fourth moment.
To ensure a constant SNR, we further set $\sigma_d = d^{1/2}\sigma_0$ for some $\sigma_0 > 0$ that is independent of $(n, d)$. 
We also assume $n, d \to \infty$ with $n / d \to \delta \in (0, \infty)$.

With these assumptions, the asymptotic spectral distribution of $\bX^{\top} \bX / n$ is characterized by the Marchenko–Pastur law \cite{marchenko1967distribution}, and the extreme eigenvalues are given by the Bai-Yin's law \cite{bai1993limit}. 
We state these two laws below for readers' convenience.
 
%
	The Marchenko–Pastur law $F_{\delta}$ has a density function
	\begin{align*}
		p_{\delta} (y) = \left\{ \begin{array}{ll}
			\frac{\delta}{2\pi x} \sqrt{(b - x)(x - a)}, & \mbox{if  } a \leq x \leq b, \\
			0, & \mbox{otherwise, }
		\end{array} \right.
	\end{align*}
	and has a point mass $1 - \delta$ at the origin if $\delta \in (0, 1)$. In the above display, $a = (1 - 1 / \sqrt{\delta})^2$ and $b = (1 + 1 / \sqrt{\delta})^2$. 
	Under the assumptions of this part, the empirical spectral distribution\footnote{The empirical spectral distribution of an $n \times n$ symmetric matrix $\bM$ refers to uniform distribution over all eigenvalues of $\bM$. } 
	of $\bX^{\top} \bX / n$ converges to $F_{\delta}$. 
	In addition, by the Bai-Yin's law, it holds that
	\begin{align}
	\label{eq:BY-law}
		\lambda_{\max}(\bX^{\top} \bX / n) \overset{a.s.}{\to} (1 + 1 / \sqrt{\delta})^2, \qquad \lambda_{\min}(\bX^{\top} \bX / n) \overset{a.s.}{\to} (1 - 1 / \sqrt{\delta})^2\mathbbm{1}\{\delta \geq 1\}. 
	\end{align}
If $\delta$ is large and $\gamma$ is only slightly larger than $\lambda_{\max} (\sigma_d^{-2} \bX^{\top} \bX)$, 
then by \cref{eq:BY-law}, the denominator on the left-hand-side of \cref{eq:convex-condition} is approximately $4 \sigma_0^{-2}\sqrt{\delta}$, while the right-hand-side of \cref{eq:convex-condition} is $O(\mathrm{polylog}(\delta) / \delta)$. 
This suggests that condition \eqref{eq:convex-condition} is satisfied with high probability for a  sufficiently large $\delta$, 
implying that the problem is feasible. 

We next characterize the asymptotic feasible region. 
Specifically, we say a parameter choice $(\delta, q, \mu, \sigma_0)$ is \emph{asymptotically feasible} 
if there exists $\gamma > 0$, such that 
\begin{align}
\label{eq:asymp-feasible}
	 \frac{1}{\gamma - \delta \sigma_0^{-2}(1 - 1 / \sqrt{\delta})^2 \mathbbm{1}\{\delta \geq 1\}} >- \inf_{x \in \RR} V_{\gamma}''(x), \qquad \gamma > \frac{\delta(1 + 1 / \sqrt{\delta})^2}{ \sigma_0^2}.   
\end{align} 
In the first equation above, the left-hand side represents the limit of $\lambda_{\min}(\bA^{-1})$ as $n, d \to \infty$, while the right-hand side indicates the maximum of $-V_{\gamma}''$. 
In the second equation, the lower bound is the limiting value of $\lambda_{\max}(\sigma_d^{-2} \bX^{\top} \bX)$.

As an illustration, we plot the asymptotic feasible regions for two diffuse densities: Gaussian and Laplace.  
Specifically, we consider $\mu = \normal(0, 1)$ and $\mu = \Laplace (\sqrt{2})$, both having unit second moments. 
We display the asymptotic feasible regions for these two diffuse densities according to \cref{eq:asymp-feasible}. 
The asymptotic feasible region for $\mu = \normal(0, 1)$ is given in Figure \ref{fig:Gaussian_regime}, and that for $\mu = \Laplace (\sqrt{2})$ is given in Figure \ref{fig:Laplace_regime}. 

In the next theorem, we show that \cref{eq:asymp-feasible} provides an almost necessary and sufficient condition for our proposed algorithm to be asymptotically feasible.

\begin{theorem}
	\label{thm:asymp-feasible}
	We assume the conditions listed in this example.
	If \cref{eq:asymp-feasible} holds for some $\gamma$, then with probability $1 - o_n(1)$, there exists $\gamma > \sigma_d^{-2} \|\bX\|_{\op}^2$, such that $\bvarphi$ given in \cref{eq:pi} has a strongly log-concave marginal distribution. 
	Namely, the problem is feasible. 
	On the other hand, if for all $\gamma \geq {\delta(1 + 1 / \sqrt{\delta})^2}{ \sigma_0^{-2}}$, it holds that 
	\begin{align*}
		\frac{1}{\gamma - \delta \sigma_0^{-2}(1 - 1 / \sqrt{\delta})^2 \mathbbm{1}\{\delta \geq 1\}} <- \inf_{x \in \RR} V_{\gamma}''(x), 
	\end{align*}
	then with probability $1 - o_n(1)$, there does not exist $\gamma > \sigma_d^{-2} \|\bX\|_{\op}^2$, such that $\bvarphi$ has a log-concave marginal distribution. 
\end{theorem}

\begin{proof}[Proof of \cref{thm:asymp-feasible}]
	We prove \cref{thm:asymp-feasible} in Appendix \ref{sec:proof-thm:asymp-feasible}. 
\end{proof}




\begin{figure}
	\centering
	\includegraphics[width=\linewidth]{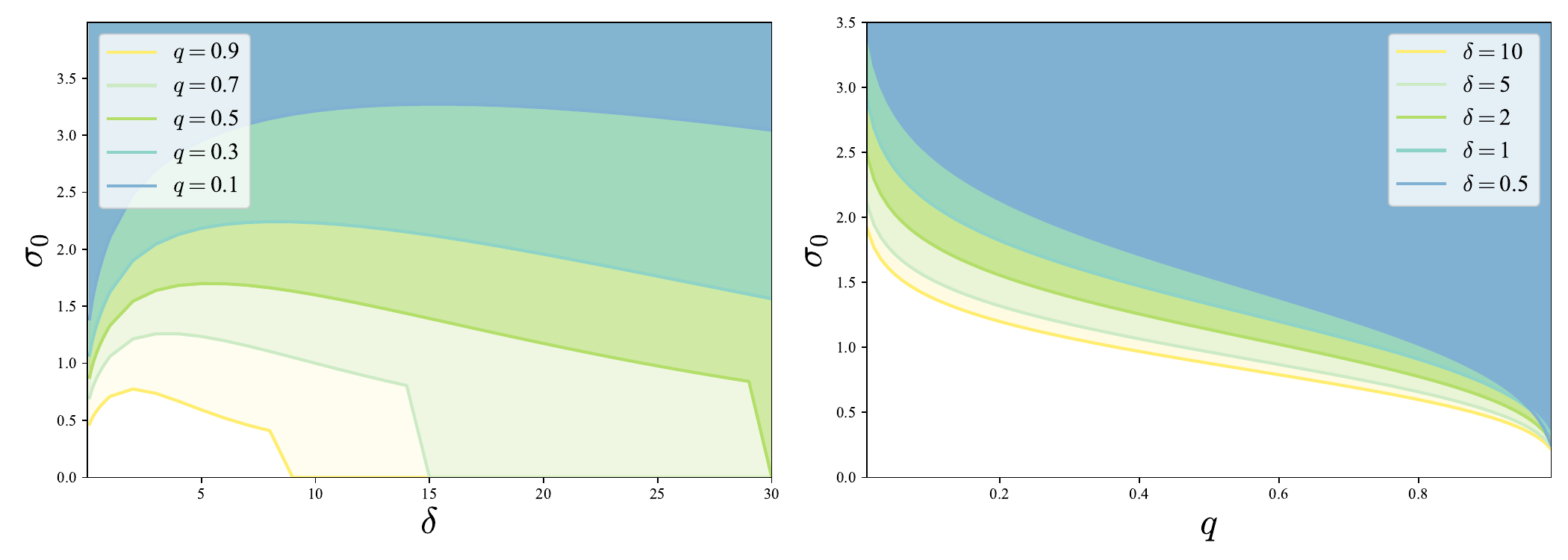}
	\caption{Illustration of the asymptotic feasible regions when $\mu = \normal(0,1)$.  
	In the above two panels, each line separates the entire panel into two regions: the colored
	 regions are asymptotically feasible, while the blank regions  are asymptotically infeasible.
	In the left panel, we fix $q \in \{0.1, 0.3, 0.5, 0.7, 0.9\}$, and use different colors to 
	indicate different values of $q$. We plot the asymptotic feasible regions in the $\delta$ -- $\sigma_0$ plane. 
	In the right panel, we fix $\delta \in \{10, 5, 2, 1, 0.5\}$, with different colors indicating
	 different choices of $\delta$, and
	we present the asymptotic feasible regions in the $q$ -- $\sigma_0$ plane. 
	 }
	\label{fig:Gaussian_regime}
\end{figure}



\begin{figure}
	\centering
	\includegraphics[width=\linewidth]{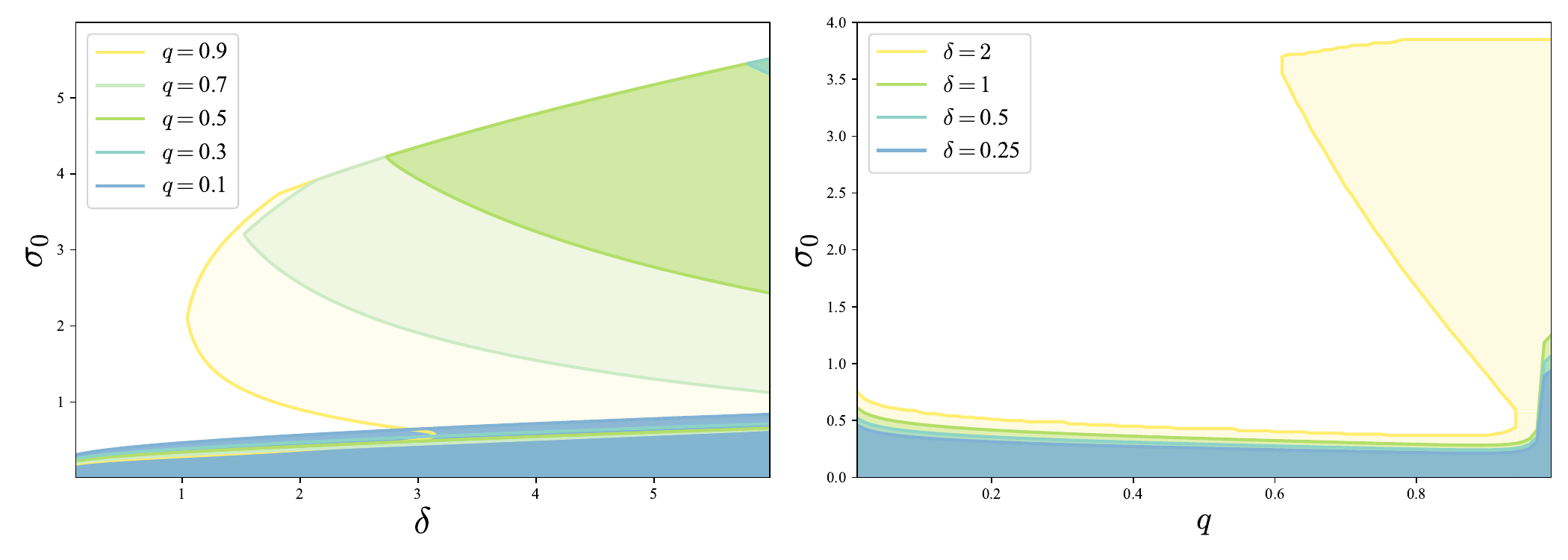}
	\caption{Illustration of the asymptotic feasible regions when $\mu = \Laplace(\sqrt{2})$.  
	In the above two panels, each line separates the entire panel into two regions: the colored regions are asymptotically feasible, while the blank regions  are asymptotically infeasible.
	In the left panel, we fix $q \in \{0.1, 0.3, 0.5, 0.7, 0.9\}$, and use different colors to indicate different $q$ values. We plot the asymptotic feasible regions in the $\delta$ -- $\sigma_0$ plane. 
	In the right panel, we fix $\delta \in \{2, 1, 0.5, 0.25\}$, and different colors stand for different choices of $\delta$. 
	We present the asymptotic feasible regions in the $q$ -- $\delta_0$ plane instead. }
	\label{fig:Laplace_regime}
\end{figure}

\section{Numerical experiments}
\label{sec:experiments}

In this section, we demonstrate the effectiveness of the proposed two-stage sampling algorithm.
We emphasize that the objective of this simulation study is not to illustrate our algorithm's ability for conducting variable selection or estimation, as these are primarily influenced by the quality of the prior distribution. 
Alternatively, we aim to validate that, given a prior distribution of the form \eqref{eq:mixture-prior} and a pair of observations $(\by, \bX)$, our algorithm is able to sample from the corresponding posterior distribution with high accuracy.
To this end, we conduct simulations on synthetic datasets.

This section is organized as follows. 
In \cref{sec:prior-arts}, we discuss several earlier algorithms for Bayesian regression
that we adopt as baselines. 
In \cref{sec:log-concave}, we review several algorithms that have been proved successful for log-concave sampling. 
We state our simulation settings in \cref{sec:simulation-settings}. 
We present implementation details of our algorithm in \cref{sec:implementation-details}. 
Finally, we present our numerical results in Section \ref{sec:simulation-outcomes}.

\subsection{Baseline algorithms}
\label{sec:prior-arts}

We summarize several sampling algorithms that tackle the spike-and-slab prior,
since we will compare them empirically with our approach.


\paragraph{Stochastic search variable selection (SSVS).}
This algorithm was first proposed in \cite{george1993variable} to handle priors that have a Gaussian mixture form.  
SSVS is a Gibbs sampler that alternates between conditional sampling for different parameters. 
Following the discussions in Section 5.2 of \cite{bai2021spike}, SSVS can also be applied to handle the spike-and-slab Lasso prior if we treat the Laplace distribution as a scale mixture of Gaussians with an exponential mixing distribution. 
The original SSVS algorithm involves computing matrix inversions and hence is computationally expensive in high-dimensional settings. 
Algebraic tricks have been proposed to reduce the computational burden, see \cite{bhattacharya2016fast,narisetty2018skinny} and the discussions in \cite{nie2023bayesian}.

\paragraph{Weighted Bayesian Bootstrap (WBB).}
Following the idea of the weighted likelihood bootstrap \cite{newton1994approximate}, the WBB 
algorithm introduced in \cite{newton2018weighted} constructs a randomly weighted posterior 
distribution by randomly assigning the observations with independently distributed weights. 
In their approach, both the likelihood and the prior values are reweighted.
They then propose to employ off-the-shelf optimization techniques to compute the mode of this reweighted posterior distribution.
This entire procedure is then independently repeated with different weight values, and the solutions to the optimization problems form a collection of samples that approximate the target posterior distribution.
Note that an optimization procedure is required to generate every single output sample.

\paragraph{Bayesian Bootstrap spike-and-slab LASSO (BB-SSL). } 
The BB-SSL approach was first proposed in \cite{nie2022bayesian}, and follows a similar idea as WBB. 
The major distinction is that instead of reweighting the priors, BB-SSL  applies random perturbations
 to the prior means.   
This method leverages the mode detection ability and efficiency of the spike-and-slab Lasso procedure
  \cite{rovckova2018spike}. 
As for the sampling algorithm, they propose to create multiple independently perturbed datasets and
 approximate the target posteriors by performing MAP optimization separately on each of the perturbed dataset.

\subsection{Log-concave sampling}
\label{sec:log-concave}

We summarize prominent examples of  sampling algorithms that come with provable guarantees when the target distribution is log-concave. 

\paragraph{Unadjusted Langevine algorithm (ULA) and underdamped
Langevin MCMC. } 
The unadjusted Langevine algorithm (ULA) is an MCMC method that updates the current state at each step using the gradient of the logarithmic density and additive Gaussian noise. 
This update mechanism aims to mimic the Langevin dynamics. 
A simple variant of ULA is the underdamped Langevin MCMC that incorporates an extra momentum term. 
When the target distribution is log-concave, upper bounds on the mixing times for both algorithms have been established. 
See, for instance, \cite{durmus2019high} for a result on ULA and \cite{cheng2018underdamped} for a result on underdamped Langevin MCMC.

\paragraph{Metropolis-adjusted Langevin algorithm (MALA). } 
The Metropolis-adjusted Langevin algorithm (MALA) differs from ULA by incorporating an accept-reject correction step.  
We refer to \cref{alg:MALA} for more details on MALA. 
The correction step ensures that the output distribution of MALA converges to the target distribution as the number of steps tends to infinity. 
For an upper bound on MALA's mixing time, see \cite{dwivedi2019log}.

\paragraph{Hamiltonian Monte Carlo (HMC). }
Hamiltonian Monte Carlo (HMC) is a powerful MCMC algorithm that leverages concepts from physics. 
At each iteration, HMC updates both the state location and a velocity term.
HMC demonstrates outstanding performance and faster convergence in many settings. 
Theoretical guarantees for HMC can be found in \cite{chen2019optimal}. 
We present implementation details of HMC in \cref{alg:HMC}.

\subsection{Simulation settings}
\label{sec:simulation-settings}

We adopt two settings, one with a Gaussian diffuse density and the other with a Laplace diffuse density. 
For this experiment, we use design matrices $\bX$ generated from a random ensemble. 
We will consider  choices of the parameters that fall both inside and outside the feasible regions 
indicated in Figures \ref{fig:Gaussian_regime} and \ref{fig:Laplace_regime}.

\subsubsection*{Setting I: Gaussian diffuse density}

In our first experiment, we let $\mu = \normal(0,1)$. 
As for the other parameters, we set $n = 100$, $d = 50$, $q = 0.2$, and $\sigma_d = 3\sqrt{d}$. 
We generate the linear coefficients $\btheta$ following the Gaussian spike-and-slab prior: 
$\theta_i \iidsim (1 - q) \delta_0 + q \, \normal(0, 1)$ for $i \in [d]$. We assume that the 
rows of $\bX$  are independent Gaussian  $\bx_i\sim\normal(\mathbf{0}_d, \bSigma)$, with 
$\Sigma_{ij} = \Sigma_{ji} = \rho^{-|i - j|}$ and $\rho \in [0, 1]$. 
We consider here $\rho \in \{0, 0.3, 0.6, 0.9\}$, and make the convention that $0^0 = 1$.  
Finally, we generate the response vector by taking $\by = \bX \btheta + \bepsilon$ for 
$\bepsilon \sim \normal(\mathbf{0}_n, \sigma_d^2)$. 

\subsubsection*{Setting II: Laplace diffuse density}

In our second experiment, we choose $\mu = \Laplace(\sqrt{2})$,
so that $\mu$ has unit second moment. 
For this experiment, we let $n = 100$, $d = 30$, $q = 0.7$, and $\sigma_d = 3\sqrt{d}$.
We assume that $\theta_i \iidsim (1 - q) \delta_0 + q \Laplace(\sqrt{2})$ for $i \in [d]$.
Once again, we assume that the features $\bx_i$ are independently generated from $\normal(\mathbf{0}_d, \bSigma)$, with $\Sigma_{ij} = \Sigma_{ji} = \rho^{-|i - j|}$, and we let $\by = \bX \btheta + \bepsilon$ for $\bepsilon \sim \normal(\mathbf{0}_n, \sigma_d^2)$. 
We also consider various correlation levels $\rho \in \{0, 0.3, 0.6, 0.9\}$.

We point out that, for $\rho=0$, the settings considered here are
feasible (both for Gaussian and Laplace densities), as can be checked from
 Figures \ref{fig:Gaussian_regime} and \ref{fig:Laplace_regime}.

\subsection{Implementation details}
\label{sec:implementation-details}

In this section, we provide the implementation details for the proposed two-stage sampling algorithm.
We first discuss sampling of $\bvarphi$, which has a log-concave density function when the parameters are feasible.

We employ two approaches: the Metropolis-adjusted Langevin algorithm (MALA) and the Hamiltonian Monte Carlo
 (HMC) algorithm. 
Throughout the experiment, we take $\gamma = \sigma_d^{-2} \|\bX\|_{\op}^2 + 0.1$ to ensure $\bA$ is positive semi-definite.

\subsubsection{Implementation details for MALA}

We consider MALA equipped with the Euler–Maruyama discretization scheme, as outlined in \cref{alg:MALA}. 
Specifically, MALA takes as inputs a positive step size $\tau > 0$ and a total number of steps $K \in \NN_+$. 
Following the suggestions from \cite{dwivedi2019log}, we initialize the MALA algorithm randomly with
 distribution $\normal(\bvarphi^{\ast}, L^{-1} \id_d)$. 
Here, $\bvarphi^{\ast}$ denotes the unique mode of the density function of $\bvarphi$, which is also the unique maximizer of $-H(\bvarphi)$ (recall that $-H(\bvarphi)$ is strongly concave). 
We also assume that $H(\bvarphi)$ is $L$-smooth, in the sense that
\begin{align*}
	H(\bvarphi_1) - H(\bvarphi_2) - \nabla H(\bvarphi_2)^{\sT} (\bvarphi_1 - \bvarphi_2) \leq \frac{L}{2} \|\bvarphi_1 - \bvarphi_2\|_2^2. 
\end{align*}
Since $\gamma = \sigma_d^{-2} \|\bX\|_{\op}^2 + 0.1$, we conclude that $\|\bA^{-1}\|_{\op} \leq 10$. 
In addition, note that the Hessian of $H(\bvarphi)$ is positive semi-definite, and is the sum of
 $\bA^{-1}$ and a diagonal matrix that has non-positive entries. 
Therefore, we conclude that $\|\nabla^2 H(\bvarphi) \|_{\op} \leq 10$. 
That is to say,  we can always take $L = 10$.

When implementing MALA, we tune the step size $\tau$ to get an acceptance rate between $30\%$ and $50\%$. 
We estimate $\bvarphi^{\ast}$ using gradient ascent\footnote{For this part, we adopt a learning rate $0.01$
and a maximum number of iterations $10^5$.}.
We also discard samples from the burn-in period. 
We determine the lengths of the burn-in period and the MCMC chain via  diagnostic plots. 
More details can be found in Appendix \ref{sec:diagnostics}.

\begin{algorithm}
\caption{Metropolis-adjusted Langevin algorithm (MALA)}\label{alg:MALA}
\begin{algorithmic}[1]
\REQUIRE Step size $\tau$, number of steps $K$;
\STATE Get an estimate of $\bvarphi^{\ast}$ via  gradient ascent, and denote it by $\widehat\bvarphi^{\ast}$; 
\STATE Initialize MALA at $\bvarphi_0 \sim \normal(\hat\bvarphi^{\ast}, L^{-1} \id_d)$, with $L = 10$; 
\FOR{$k = 1, 2, \cdots, K$}
	\STATE $\mathtt{Proceed} \gets \mathtt{False}$;
	\WHILE{$\mathtt{Proceed} = \mathtt{False}$}
	\STATE Generate $\xi_k \sim \normal(\mathbf{0}, \id_d)$ independent of everything else so far; 
	\STATE $\bvarphi_{k}' \gets \bvarphi_{k - 1} - \tau \cdot  \nabla H(\bvarphi_{k - 1}) + \sqrt{2\tau} \xi_k$; 
	\STATE $\alpha \gets \min \left\{ 1,\, e^{-H(\bvarphi_{k}') + H(\bvarphi_{k - 1})}\cdot \frac{q( \bvarphi_{k - 1} \mid \bvarphi_{k}' )}{q( \bvarphi_{k}' \mid \bvarphi_{k - 1})} \right\}$, where
	$$
		q(\bx' \mid \bx) \propto \exp\left( -\frac{1}{4\tau} \|\bx' - \bx + \tau \nabla H(\bx)\|^2 \right); 
$$
	\STATE Sample $u \sim \Unif[0,1]$;
	\IF{$u \leq \alpha$}
		\STATE $\mathtt{Proceed} \gets \mathtt{True}$; 
		\STATE $\bvarphi_k \gets \bvarphi_k'$;
	\ENDIF
	\ENDWHILE

\ENDFOR
\end{algorithmic}
{\bf Return}: $\{\bvarphi_k: k \in [K]\}$.
\end{algorithm}

\subsubsection{Implementation details for HMC}

Alternatively, we can apply HMC to sample $\bvarphi$, which we state as \cref{alg:HMC} in the appendix. 
At each step, HMC proposes to update the current step following the outcome of a leapfrog integrator. 
Specifically, HMC requires inputting the mass matrix $\bOmega$, the leapfrog stepsize $\epsilon$, the number of leapfrog steps $\ell$, and the Monte Carlo steps $K$. 
In this experiment, we set $\bOmega = \id_d$, and adjust the other parameters based on diagnostic plots. More details of the diagnostic step can be found in Appendix \ref{sec:diagnostics}.   
Similar to the MALA setting, we initialize HMC at $\normal(\bvarphi^{\ast}, L^{-1} \id_d)$ with $L = 10$, 
and set $\gamma = \sigma_d^{-2} \|\bX\|_{\op}^2 + 0.1$. 
%

\subsubsection{Algorithm pipeline}

Given $\bvarphi$, we can then sample $\btheta$ from the conditional distribution $\pi (\dd \btheta \mid \bvarphi, \by, \bX)$, 
which per \cref{eq:conditional-product} has a product form and is easy to sample.

For the sampling of $\bvarphi$, in this experiment we utilize one of MALA and HMC. 
After a burn-in period, for each $\bvarphi$ sample in the Markov chain, we will sample $\btheta$ from the product conditional distribution $\pi(\de \btheta \mid \bvarphi, \by, \bX)$. 
This procedure is detailed in \cref{alg:theta}. 
\begin{algorithm}
\caption{Sampling $\btheta$}\label{alg:theta}
\begin{algorithmic}[1]
\REQUIRE Number of burn-in steps $B$, number of desired samples $N$, an MCMC sampler for $\bvarphi$; 
\STATE Implement the MCMC sampler for $\bvarphi$ and discard the first $B$ samples from the burn-in period;
\STATE $\cS \gets \emptyset$;
\FOR{$i = 1, 2, \cdots, N$}
	\STATE Perform one update step using the given MCMC sampler, and get a new sample $\bvarphi$; 
	\STATE Sample $\btheta_i \sim \pi(\btheta \mid \bvarphi, \by, \bX)$; 
	\STATE $\cS \gets \cS \cup \{\btheta_i\}$;
\ENDFOR
\end{algorithmic}
{\bf Return}: $\cS$
\end{algorithm}

\subsection{Simulation outcomes}
\label{sec:simulation-outcomes}

We will evaluate the quality of samples produced by \cref{alg:theta} by assessing their effectiveness in performing uncertainty quantification. 


Specifically, we consider the two settings listed in \cref{sec:simulation-settings}.
For every realization of $(\btheta, \bX, \by)$, we  run our proposed two-stage algorithm based on 
MALA or HMC as well as several other sampling algorithms listed in \cref{sec:prior-arts}. 
To set up comparison, we also draw samples using the Python library pymc3, which builds
 on advanced sampling techniques such as the No-U-Turn Sampler (NUTS). 
For the implementation of SSVS, WBB, and BB-SSL, we use the R package BBSSL developed by
 Nie and Rockova \cite{nie2023bayesian}. 
BBSSL is specifically designed to handle the spike-and-slab LASSO prior. 
To apply BBSSL with a point-mass spike-and-slab prior in our context,
 we specify a sufficiently small variance for the spike. 
Additionally, we adopt a separable penalty and input the true mixture probabilities into
 the function calls. 
 
We consider both correctly specified prior and incorrectly specified prior when implementing 
the proposed two-stage algorithm and pymc3. 
Here, we say the prior is incorrectly specified if we run the algorithm assuming 
$\mu = \normal(0,1)$ when $(\btheta, \bX, \by)$ is generated following the prior distribution of setting II, 
or we run the algorithm assuming $\mu = \Laplace(\sqrt{2})$ when $(\btheta, \bX, \by)$ is generated 
following the prior distribution of setting I. 

For each algorithm, we collect $10^4$ samples after the burn-in period (the length of the 
burn-in period is discussed in Appendix \ref{sec:diagnostics}), 
and construct credible intervals for every coordinate of $\btheta$ based on the collected samples. 
Specifically, we take the 2.5\% and 97.5\% quantiles as the lower and upper ends for the credible interval.
We can then determine whether the constructed intervals contain the true coordinates by 
checking the entries of the true coefficient vector $\btheta$. 

To assess sample quality, we repeat this experiment independently for 1000 times and compute
 the empirical coverage probability.
Namely, we independently generate 1000 data tuples $(\btheta, \bX, \by)$, construct credible
 intervals, get $1000d$ coverage indicators, and take the average of these $1000d$ indicators. 
We expect this average to be around 0.95 if the algorithm adopted are indeed sampling from the target posterior. 
For both settings I and II, we display the empirical coverage rates under different design matrix settings as 
Figure \ref{fig:coverage_bar}.
The figure shows that our algorithm maintains coverage levels close to 0.95 in nearly all settings. 
However, we observe a slight overcoverage for our algorithm  when using HMC for log-concave sampling. Since MALA performs well in the same setting, we suspect the overcoverage is
  caused by HMC rather than the measure decomposition. 

\begin{figure}[ht]
	\centering
	\includegraphics[width=\linewidth]{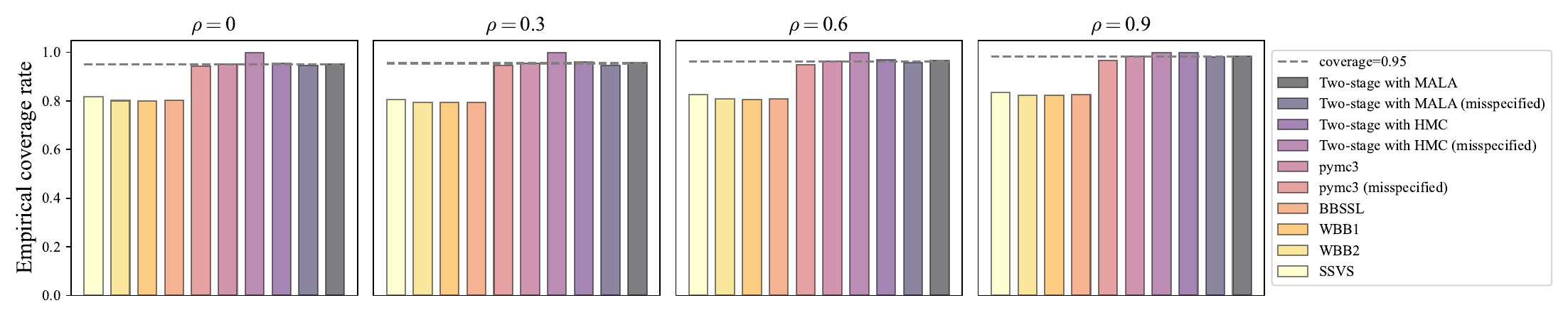}
	\includegraphics[width=\linewidth]{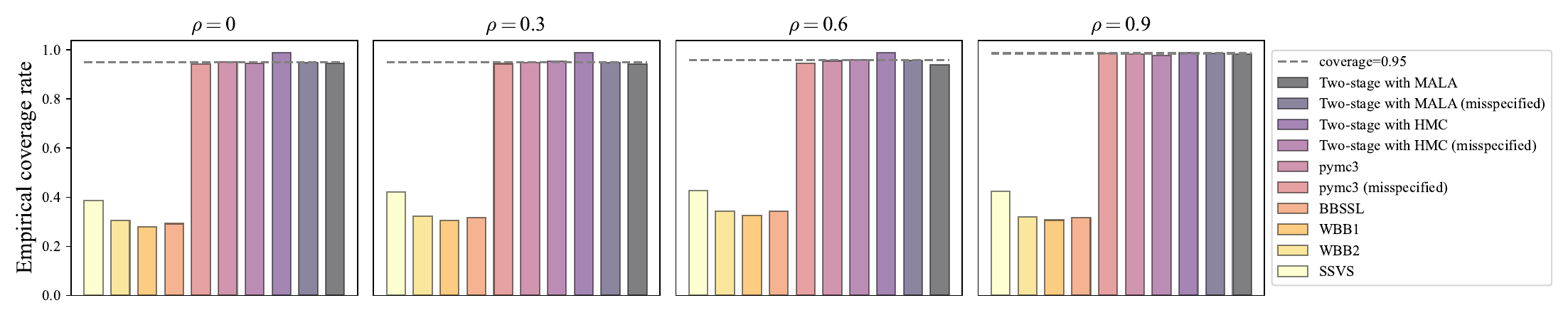}
	\caption{Bar plots that demonstrate empirical coverage rates. The upper panel is for setting I 
	and the bottom panel is for setting II. 
	For this plot, we independently conducted the experiment 1,000 times, and computed the empirical 
	coverage rates by averaging over 1,000 outcomes. Here, different bars represent the empirical coverage rates for 
	different algorithms, and the horizontal dashed line indicates the 0.95 desired coverage level. }
	\label{fig:coverage_bar}
\end{figure}

We emphasize that our algorithm offers no theoretical guarantees outside the feasible region, 
and may perform poorly compared to other algorithms in such cases. 
To illustrate this point, 
we consider a simple example with parameters $n = 5$, $d = 20$, $q = 0.2$, and $\sigma_d = 1$. 
We check two cases $\mu = \normal(0, 1)$ and $\mu = \Laplace(\sqrt{2})$. 
As for the design matrix, once again we assume $\bx_i \sim_{i.i.d.} \normal(\mathbf{0}_d, \bSigma)$ with $\Sigma_{ij} = \Sigma_{ji} = \rho^{|i - j|}$ for $i, j \in [d]$. 
We perform a similar experiment and compute the empirical coverage rates for different algorithms, assuming that our two-stage algorithm always has access to a correctly specified prior. 
We present the simulation outcomes in Figure \ref{fig:coverage_bar_outside}.
From the figure, we see that our algorithm achieves lower coverage rates than other algorithms and 
falls short of the desired 0.95 benchmark.

\begin{figure}[ht]
		\centering
	\includegraphics[width=\linewidth]{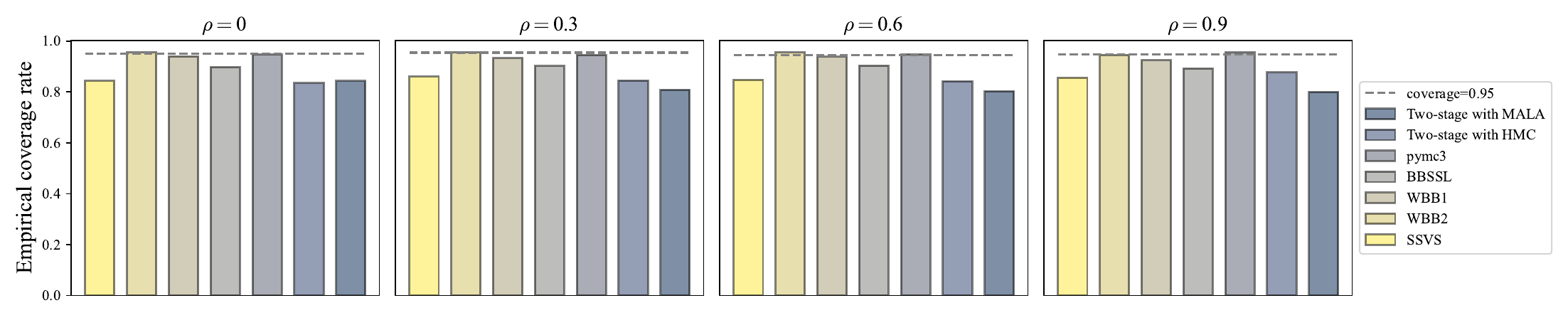}
	\includegraphics[width=\linewidth]{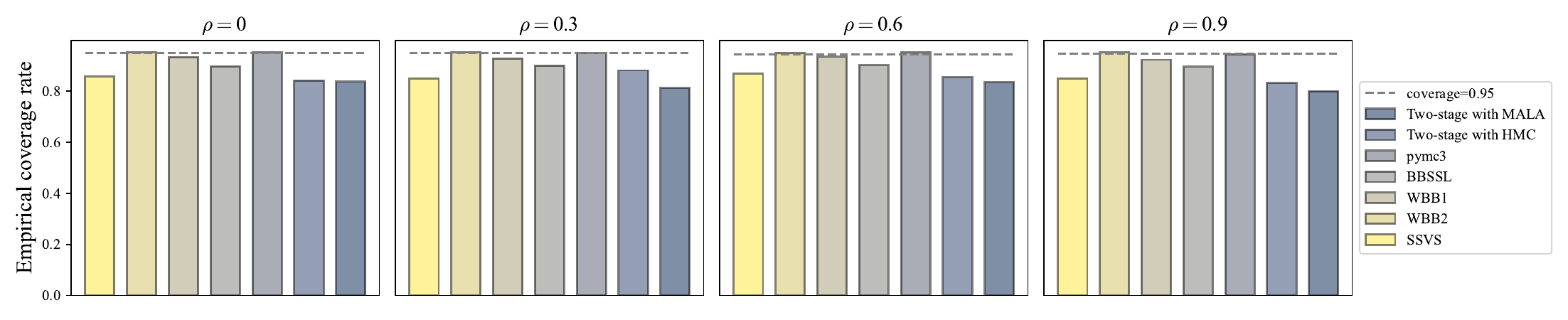}
		\caption{Bar plots that demonstrate empirical coverage rates. The upper panel is for $\mu = \normal(0, 1)$, and the bottom panel is for $\mu = \Laplace(\sqrt{2})$. 
	To make this plot, we independently conducted the experiment 1,000 times, and computed the empirical coverage rates by taking the average. Different bars stand for the empirical coverage rates for different algorithms. The horizontal dashed line is the 0.95 desired coverage level. }
	\label{fig:coverage_bar_outside}
\end{figure}

Note that the feasible region depends on $(\bX, q, \sigma_d, \mu)$. 
We observe $\bX$ as a part of the data. 
The other parameters can be estimated using methods such as empirical Bayes.
Hence, it is possible to determine from the data whether we are operating
within the feasible region or not.
In the former case, the algorithm is guaranteed to produce samples that approximate the 
correct posterior.

We emphasize that, in principle, it is possible for the algorithm to succeed also outside 
the feasible region. 
Indeed, even if the density of $\bvarphi$ is not log-concave, it might be mildly so, 
and MALA or HMC might still be able to mix rapidly.

Finally, we comment that the unsatisfactory performance of our algorithm in the context of Figure \ref{fig:coverage_bar_outside} is likely due to the mixing failure of Markov chains when the target distribution (i.e., the marginal distribution of $\bvarphi$) is not log-concave. 
Such inferior behavior remains as we increase the burn-in period length and the thinning interval length. 
Specifically, suppose we  take one sample in every $g$ samples along the $\bvarphi$ Markov chain, and set the length of the burn-in period to be $g \times 10^4$ steps. 
For $g \in \{1, 2, \cdots, 10\}$, we use the two-stage algorithm to collect $10^4$ samples, construct the credible intervals, repeat the procedure 1000 times, and compute the empirical coverage rates. 
We present the outcomes of this experiment as Figure \ref{fig:increase_length_and_gap}. 
From this figure, we see that augmenting the burn-in period and the thinning interval offers little improvement in terms of empirical coverage rates, suggesting the difficulty of sampling from non-log-concave distributions.  

\begin{figure}[ht]
\centering
	\includegraphics[width=\textwidth]{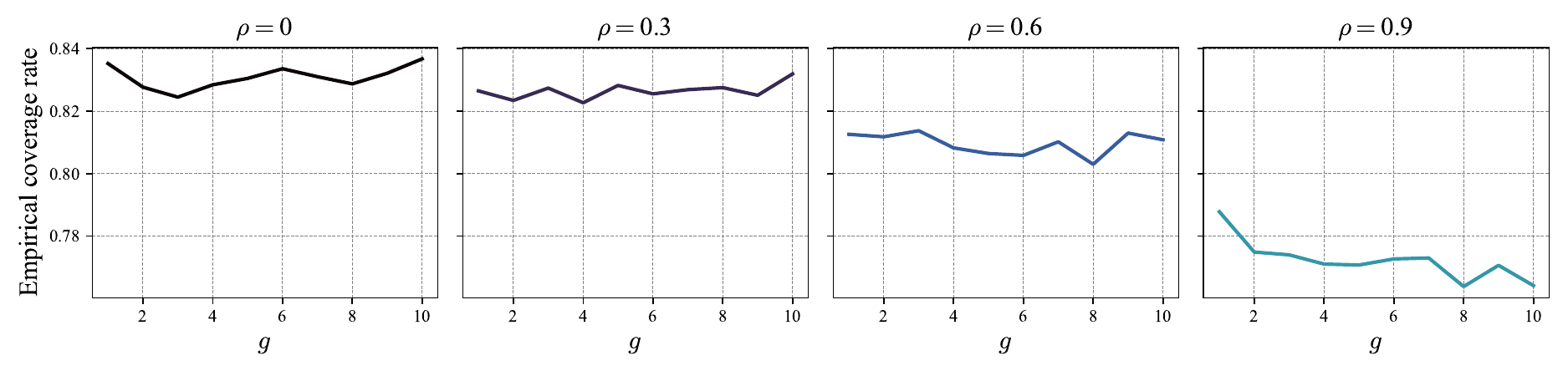}
\caption{Line charts that display the empirical coverage rates. 
Here, we adopt the same parameter setting as that of Figure \ref{fig:coverage_bar_outside}, and focus on the Gaussian mixture prior with $\mu = \normal(0,1)$. For this experiment, we implement MALA for the sampling of $\bvarphi$, and consider $\rho \in \{0, 0.3, 0.6, 0.9\}$ and $g \in [10]$. We plot the empirical coverage rates obtained from 1,000 independent experiments for different combinations of $(\rho, g)$. From the figure, we see that increasing the burn-in period and the thinning interval length does not improve the realized coverage rates produced by our algorithm in a setting that is outside the feasible region.  }
\label{fig:increase_length_and_gap}
\end{figure}

\subsection*{Acknowledgment}

This work was supported by the NSF through award DMS-2031883, the Simons Foundation through
Award 814639 for the Collaboration on the Theoretical Foundations of Deep Learning, the NSF
grant CCF2006489 and the ONR grant N00014-18-1-2729.

\bibliographystyle{alpha}
\bibliography{bib.bib}

\newpage

\begin{appendices}

\section{List of figures}

We present in this section a list of all figures that appear in our paper. 
The summary is presented as Table \ref{table:figure_list}. 

\begin{table}[!ht]
\centering
\begin{tabular}{l  l  l}
    \toprule
    \textbf{Section} & \textbf{Figure} & \textbf{Description} \\
    \midrule
    \cref{sec:introduction} & Figure \ref{fig:shape} & 
    Comparison of the output distribution and the target posterior \\
    \cref{sec:asymptotic-feasibility} & Figure \ref{fig:Gaussian_regime} & Asymptotic feasible region when $\mu = \normal(0, 1)$ \\
    \cref{sec:asymptotic-feasibility} & Figure \ref{fig:Laplace_regime} & Asymptotic feasible region when $\mu = \Laplace(\sqrt{2})$ \\
    \cref{sec:simulation-outcomes} & Figure \ref{fig:coverage_bar} & Empirical coverage rates attained by different algorithms \\
    \cref{sec:simulation-outcomes} & Figure \ref{fig:coverage_bar_outside} & Performance of two-stage algorithm outside the feasible region \\
    \cref{sec:simulation-outcomes} & Figure \ref{fig:increase_length_and_gap} & Empirical coverage rates with different chain parameters  \\
    \cref{sec:diagnostics1} & Figure \ref{fig:setting1-MALA-diagnostics} & Diagnostic plots for MALA under Setting I \\
    \cref{sec:diagnostics2} & Figure \ref{fig:setting1-HMC-diagnostics} & Diagnostic plots for HMC under Setting I \\
    \cref{sec:diagnostics3} & Figure \ref{fig:setting2-MALA-diagnostics} & Diagnostic plots for MALA under Setting II \\
    \cref{sec:diagnostics4} & Figure \ref{fig:setting2-HMC-diagnostics} & Diagnostic plots for HMC under Setting II \\
    \bottomrule
\end{tabular}
\caption{List of figures that appear in our paper. }
\label{table:figure_list}
\end{table}

\section{Proof of Lemma \ref{lemma:V-gamma-general}}
\label{sec:proof-lemma:V-gamma-general}	

\begin{proof}[Proof of Lemma \ref{lemma:V-gamma-general}]
	For simplicity, we define 
	\begin{align*}
		g(x) := \int_{\R} e^{x \theta - \frac{\gamma}{2} \theta^2} \mu(\de \theta), \qquad p(x)  := \frac{q g(x)}{1 - q + q g(x)}. 
	\end{align*}
	Note that $p(x) \in [0, 1]$ for all $x \in \RR$. 
	Taking the first and second derivatives of $V_{\gamma}$, we obtain 
	\begin{align}\label{eq:V-gamma-pp}
	\begin{split}
		 V_{\gamma}'(x) = &- \frac{q g'(x)}{1 - q + q g(x)}, \\
		 V_{\gamma}''(x) 
		 = & - p(x) \cdot \left( \frac{g''(x)}{g(x)} - \Big( \frac{g'(x)}{g(x)} \Big)^2 \right) - p(x)(1 - p(x)) \cdot \left( \frac{g'(x)}{g(x)} \right)^2. 
	\end{split}
	\end{align}
	In addition, standard calculations reveal that the first and second derivatives of $g$ are related to the conditional expectation and variance, as shown in the following equalities  
	\begin{align}\label{eq:conditional-e-and-v}
	\begin{split}
		& \frac{g'(x)}{g(x)} = \E\left[ U \mid {\gamma} U + \sqrt{\gamma} Z =  x \right],  \\
		& \frac{g''(x)}{g'(x)} - \left( \frac{g'(x)}{g(x)} \right)^2 = \Var\left[ U \mid {\gamma} U + \sqrt{\gamma} Z =  x \right].
	\end{split} 
	\end{align}
	In the above display, $U$ and $Z$ are independent random variables with marginal distributions $\mu$ and $\normal(0, 1)$, respectively. 
	Putting together \cref{eq:V-gamma-pp,eq:conditional-e-and-v}, we immediately see that $V_{\gamma}''(x) \leq 0$ for all $x \in \RR$. 
	
Next, we apply the Brascamp–Lieb inequality \cite{brascamp1976extensions} to bound $V_{\gamma}''$. 
We copy this inequality below for readers' convenience. 
\begin{lemma}[Brascamp–Lieb inequality]
\label{lemma:brascamp-lieb}
	For any distribution with density $\pi \propto \exp(-H)$ for some $H: \RR^n \to \RR$, if there is a positive semidefinite matrix $\Gamma$ such that $\nabla^2 H  \succeq \Gamma$, then $\Cov[\pi] \preceq \Gamma^{-1}$.
\end{lemma}
Recall that we have assumed $\mu$ is log-concave. 
As a consequence, the posterior distribution of $U$ given ${\gamma} U + \sqrt{\gamma}Z = x$ is also log-concave. 
In fact, the logarithmic of the posterior density is $\gamma$-strongly concave, regardless of the value of $x$. 
Using \cref{lemma:brascamp-lieb}, we conclude that  
\begin{align*}
	\frac{g''(x)}{g'(x)} - \left( \frac{g'(x)}{g(x)} \right)^2 = \Var\left[ U \mid {\gamma} U + \sqrt{\gamma} Z =  x \right] \leq \gamma^{-1}. 
\end{align*}
%
%
Recall that by assumption $\mu$ is symmetric, combining this assumption with the expression in \cref{eq:conditional-e-and-v}, we conclude that $g'(0) / g(0) = 0$. 
In addition, we have shown that 
\begin{align*}
	\frac{\dd}{\dd x} \left( \frac{g'(x)}{g(x)} \right) = \frac{g''(x)}{g'(x)} - \left( \frac{g'(x)}{g(x)} \right)^2 \in [0, \gamma^{-1}]. 
\end{align*}
Putting together these two parts, we conclude that $(g'(x) / g(x))^2 \leq \gamma^{-2}x^2$ for all $x \in \RR$.
Substituting the upper bounds we have derived into \cref{eq:V-gamma-pp}, we get 
		\begin{align}\label{eq:V''-bound}
			V_{\gamma}''(x) \geq -\gamma^{-1} - \gamma^{-2} x^2p(x) (1 - p(x)) / 2. 
		\end{align}
		When $f_{\mu}(x) \geq c_1 e^{-c_2 x^{2k}}$, it holds that
		\begin{align*}
			g(x) \geq  \int_{\RR} c_1 e^{x \theta - \frac{\gamma}{2} \theta^2 - c_2 \theta^{2k}} \de \theta \geq \int_{\RR} c_1 e^{x \theta - \frac{\gamma^k \theta^{2k}}{2k} - \frac{k - 1}{2k} - c_2 \theta^{2k}} \de \theta, 
		\end{align*}
		where to obtain the second lower bound we leverage Young's inequality. Setting $\theta = x^{1 / (2k - 1)} u$, we further obtain that 
		\begin{align}\label{eq:g(x)-lower-bound}
			g(x) \geq & c_1 |x|^{1 / (2k - 1)}\int_{\RR} \exp \left(x^{2k / (2k - 1)} \cdot ( u - \gamma^{k}u^{2k} / 2k - c_2 u^{2k}) - (k - 1) / 2k    \right) \de u.  
		\end{align}
		Note that when $(c_2 + \gamma^{k} / 2k)^{-1 / (2k - 1)} / 3 \leq u \leq 2(c_2 + \gamma^{k} / 2k)^{-1 / (2k - 1)} / 3$, it holds that
		\begin{align}
		\label{eq:a-lower-bound}
		\begin{split}
			u - \gamma^{k}u^{2k} / 2k - c_2 u^{2k} = & u \cdot \left(1 - u^{2k-1} (c_2 + \gamma^k / 2k)\right) \\
			\geq & \left((c_2 + \gamma^{k} / 2k)^{-1 / (2k - 1)} / 3 \right) \cdot \left(1 - (2/3)^{2k - 1}\right) \\
			\geq & \left(c_2 + \gamma^{k} / 2k \right)^{-1 / (2k - 1)} / 9. 
		\end{split}
		\end{align}		
		%
		%
		Plugging \cref{eq:a-lower-bound} into \cref{eq:g(x)-lower-bound}, we conclude that  
		\begin{align}
		\label{eq:g-lower-bound}
			g(x) \geq \frac{c_1|x|^{1/(2k - 1)} e^{-(k - 1) / 2k}}{3 (c_2 + \gamma^{k} / 2k)^{1/(2k - 1)}} \cdot \exp \left( \frac{x^{2k / (2k - 1)}}{9\,(c_2 + \gamma^{k} / 2k)^{1 / (2k - 1)}} \right). 
		\end{align}
		Inspecting \cref{eq:g-lower-bound}, we see that when 
		\begin{align}\label{eq:x-lower-bound}
			|x| \geq 1 + 18^{\frac{2k - 1}{2k}} \cdot \big( c_2 + \gamma^k / 2k \big)^{\frac{1}{2k}} \cdot \left\{ 2\log \left( \frac{3(c_2 + \gamma^{k} / 2k)^{1/(2k - 1)}} {c_1 q e^{-(k - 1) / 2k}}\right)  \vee \log \frac{1 - q}{q} \vee 0 \right\}^{\frac{2k - 1}{2k}},
		\end{align}
		%
		%
		we have  
		\begin{align*}
			g(x) \geq   e(x) \geq \frac{1 - q}{q}, \qquad e(x) := \exp \left( \frac{x^{2k / (2k - 1)}}{18\,(c_2 + \gamma^{k} / 2k)^{1 / (2k - 1)}} \right)  
		\end{align*}
		As a result, for all $x$ that satisfies the lower bound given in \cref{eq:x-lower-bound}, it holds that $p(x) = qg(x) / (1 - q + q g(x)) \geq 1 / 2$, and  $p(x) (1 - p(x)) \leq q(1 - q) e(x) / (1 - q + q e(x))^2$. In addition, note that
		\begin{align*}
			x^2 = \left( 18\,(c_2 + \gamma^{k} / 2k)^{1 / (2k - 1)} \cdot  \log e(x) \right)^{\frac{2k - 1}{k}}. 
		\end{align*}
		Combining the above results and \cref{eq:V''-bound}, we arrive at the following lower bound: 
		\begin{align*}
			\inf_{x \in \RR} V_{\gamma}''(x) \geq -\gamma^{-1} - 18^{\frac{2k - 1}{k}} \gamma^{-2} \left( c_2 + \gamma^{k} / 2k  \right)^{\frac{1}{k}} \cdot \left( \log e(x) \right)^{\frac{2k - 1}{k}} \cdot \frac{q(1 - q) e(x)}{(1 - q + qe(x))^2}
		\end{align*}
		%
		%
		By its definition, for all $x \in \RR$ we have $e(x) \geq 1$. 
		Furthermore, 
		 $\sup_{e \geq 1} (\log e)^{(2k - 1) / k} \cdot q(1 - q) e / (1 - q + q e)^2 < \infty$, and the maximum value is a function of $q$ only. 
		 As a consequence, we conclude that for all $|x|$ exceeding the lower bound given in \cref{eq:x-lower-bound}, there exists a constant $C_0 > 0$ depending only on $(\mu, q)$, such that $ \inf_{x \in \RR} V_{\gamma}''(x) \geq -C_0 (\gamma^{-1} + \gamma^{-2})$. On the other hand, for all $x$ that does not satisfy \cref{eq:x-lower-bound}, plugging the upper bound for $|x|$ into \cref{eq:V''-bound} gives 
		\begin{align*}
			& \inf_{x \in \RR} V_{\gamma}''(x) \\
			\geq & -\gamma^{-1} - \gamma^{-2} \cdot \left(1 +  18^{\frac{2k - 1}{k}}  (c_2 + \gamma^{k} / 2k)^{\frac{1}{k}} \cdot \left\{ 2\log \Big( \frac{3(c_2 + \gamma^{k} / 2k)^{1/(2k - 1)}} {c_1 q e^{-(k - 1) / 2k}}\right)  \vee \log \frac{1 - q}{q} \vee 0 \right\}^{\frac{2k - 1}{k}} \Big). 
		\end{align*}
		In the above display, note that the constant in the parentheses that follows $\gamma^{-2}$ depends only on $(\mu, q)$.
		Therefore, there exists a constant $C_0 > 0$ that is a function of $(\mu, q)$ only, such that for all $x$ that does not satisfy \cref{eq:x-lower-bound}, it holds that $ \inf_{x \in \RR}V_{\gamma}''(x) \geq -C_0(\gamma^{-1} + \gamma^{-2}) \cdot (1 + \log (\gamma + 1))^{\frac{2k - 1}{k}}$. The proof is complete.

\end{proof}

\section{Diagnostics}
\label{sec:diagnostics}

We present in this section diagnostic plots that guide the parameter selection for MALA and HMC. 
\subsection{Setting I, MALA}
\label{sec:diagnostics1}

For the MALA implementation under setting I, we take the number of burn-in steps $B$ to be $10^4$ for $\rho \in \{0, 0.3, 0.6\}$ and $B = 2 \times 10^4$ for $\rho = 0.9$. 
By looking at the trace plots for $\bvarphi$ under different choices of $\rho$, we see that these numbers are sufficient for MALA to mix. 
As for the MALA step size, we take $\tau = 0.2$, which results in acceptance rates between 20\% and 50\% for all $\rho$ between 0.0 and 0.9. We present the diagnostic plots in Figure \ref{fig:setting1-MALA-diagnostics}.  

\begin{figure}[ht]
    \centering
     \begin{minipage}{0.31\textwidth}
        \centering
        \includegraphics[width=\textwidth]{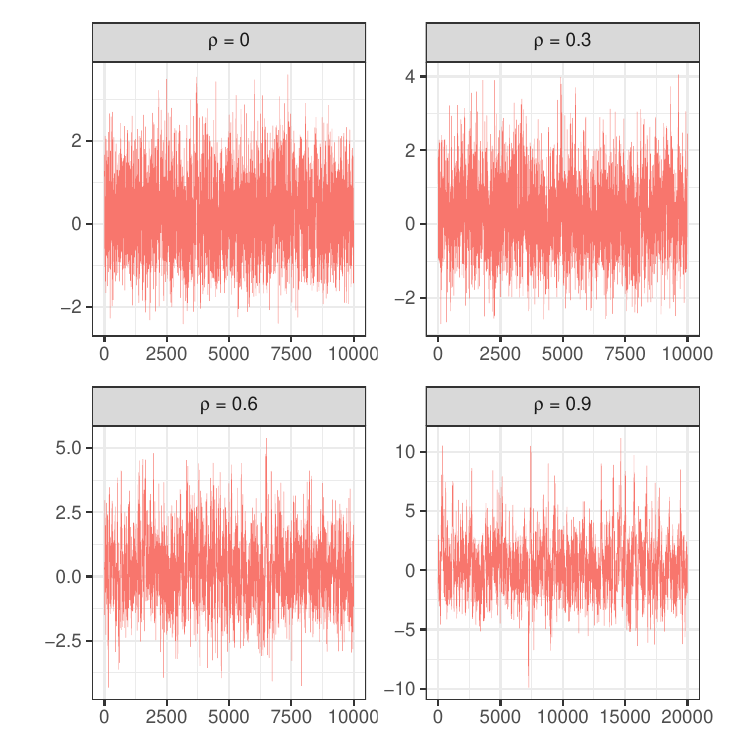} 
    \end{minipage} \hfill 
    \begin{minipage}{0.27\textwidth}
        \centering
        \includegraphics[width=\textwidth]{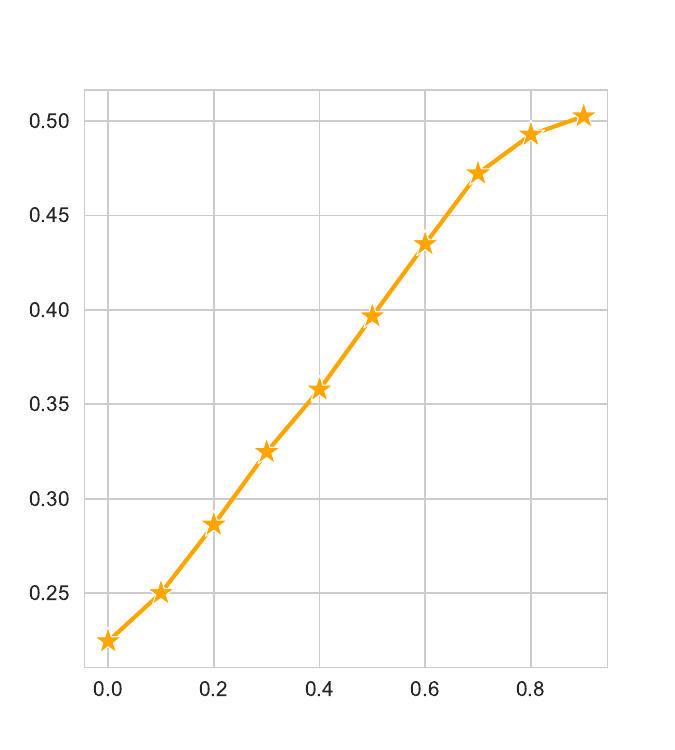} 
    \end{minipage} \hfill 
    \begin{minipage}{0.4\textwidth}
        \centering
        \includegraphics[width=\textwidth]{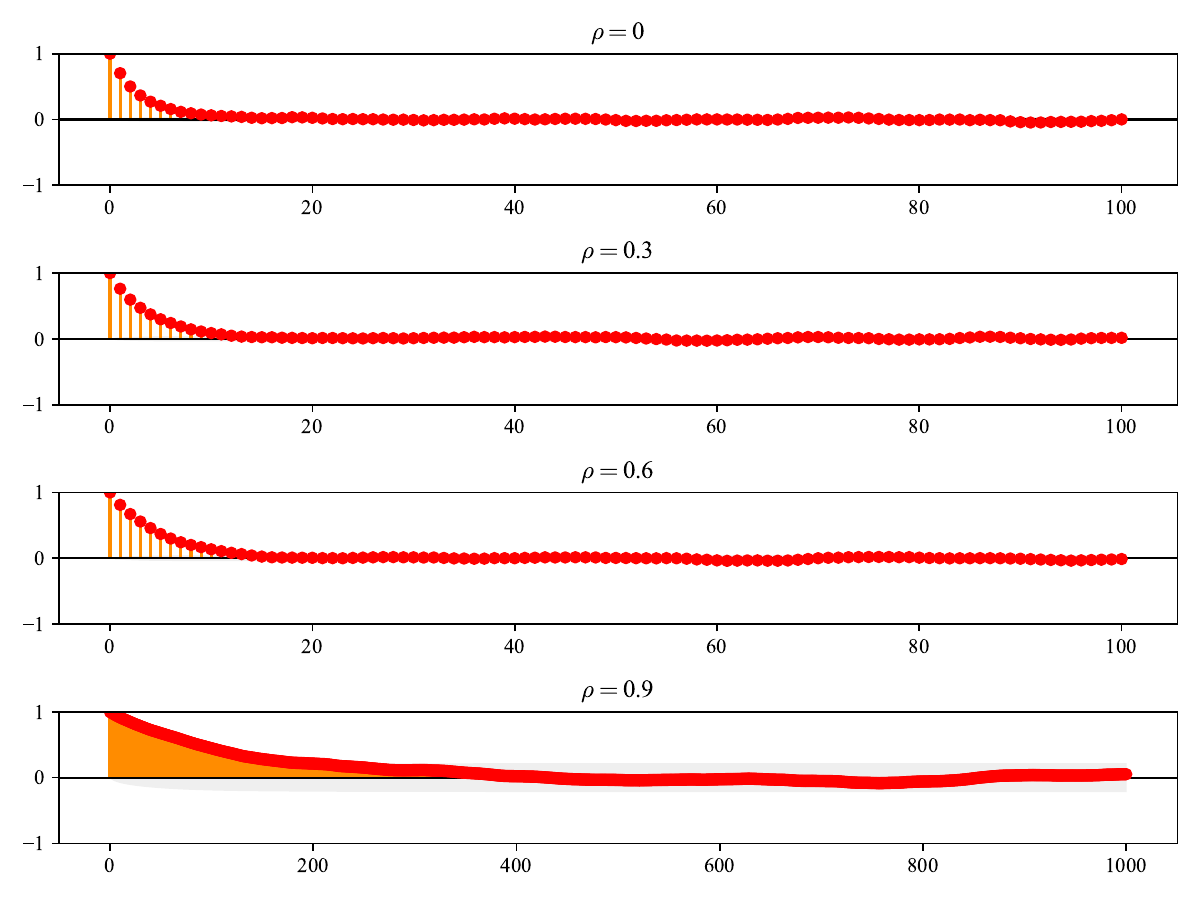}        
    \end{minipage}
    \caption{Diagnostic plots for MALA. Data is generated according to setting 1, and the sampling algorithm follows that stated in \cref{sec:diagnostics1}. Left panel: trace plots for a randomly selected coordinate in a single realization, for $\rho \in \{0, 0.3, 0.6, 0.9\}$. Middle panel: MALA acceptance rate for $\rho$ between $0$ and $0.9$. Right panel: autocorrelation plots for $\rho \in \{0, 0.3, 0.6, 0.9\}$.  }
    \label{fig:setting1-MALA-diagnostics}
\end{figure}

\subsection{Setting I, HMC}
\label{sec:diagnostics2}

To implement HMC under setting I, we set $\bOmega = \id_d$, $\epsilon = 0.4$ and $\ell = 10$. Once again, we take $B = 10^4$ for $\rho \in \{0, 0.3, 0.6\}$ and $B = 2 \times 10^4$ for $\rho = 0.9$. The diagnostic plots can be found in Figure \ref{fig:setting1-HMC-diagnostics}. 
 Note that in practice, HMC typically achieves a higher acceptance rate than MALA, with the desired acceptance rate ranging between 80\% and 99\%. 
\begin{figure}[ht]
    \centering
     \begin{minipage}{0.31\textwidth}
        \centering
        \includegraphics[width=\textwidth]{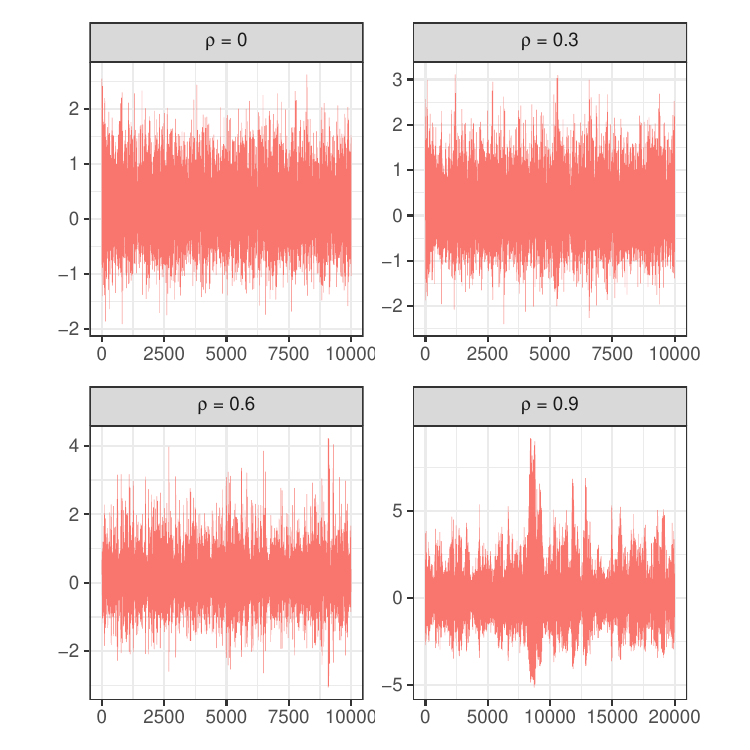} 
    \end{minipage} \hfill 
    \begin{minipage}{0.27\textwidth}
        \centering
        \includegraphics[width=\textwidth]{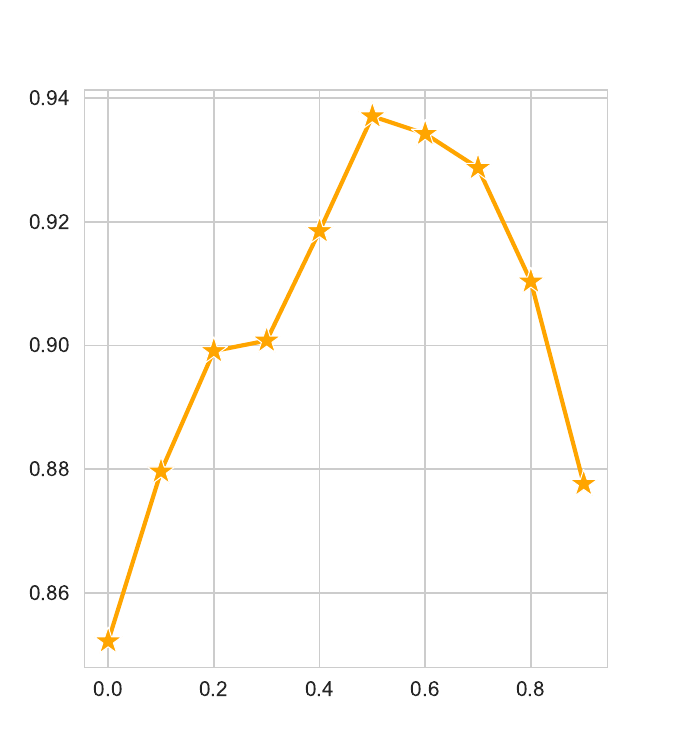} 
    \end{minipage} \hfill 
    \begin{minipage}{0.4\textwidth}
        \centering
        \includegraphics[width=\textwidth]{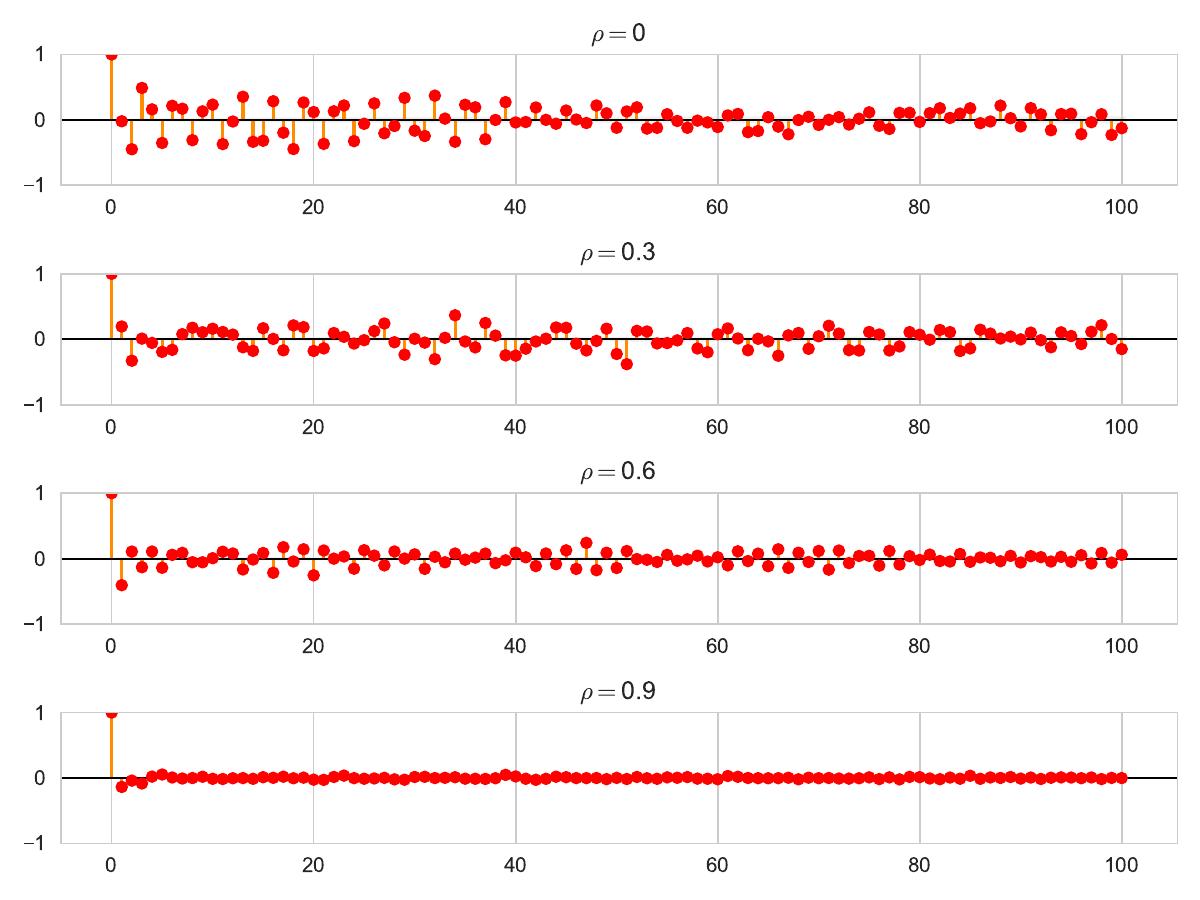}        
    \end{minipage}
    \caption{Diagnostic plots for HMC. Data is generated following setting 1, and the HMC parameters follow that stated in \cref{sec:diagnostics2}. Left panel: trace plots for a randomly selected coordinate in a single realization, for $\rho \in \{0, 0.3, 0.6, 0.9\}$. Middle panel: HMC acceptance rate for $\rho$ between $0$ and $0.9$. Right panel: autocorrelation plots for $\rho \in \{0, 0.3, 0.6, 0.9\}$. }
    \label{fig:setting1-HMC-diagnostics}
\end{figure}

\subsection{Setting II, MALA}
\label{sec:diagnostics3}

We then switch to setting II. For MALA, we take $\tau = 0.2$, $B = 10^4$ for $\rho \in \{0, 0.3, 0.6\}$ and $B = 2 \times 10^4$ for $\rho = 0.9$. We collect the diagnostic plots in Figure \ref{fig:setting2-MALA-diagnostics}. 

\begin{figure}[ht]
    \centering
     \begin{minipage}{0.31\textwidth}
        \centering
        \includegraphics[width=\textwidth]{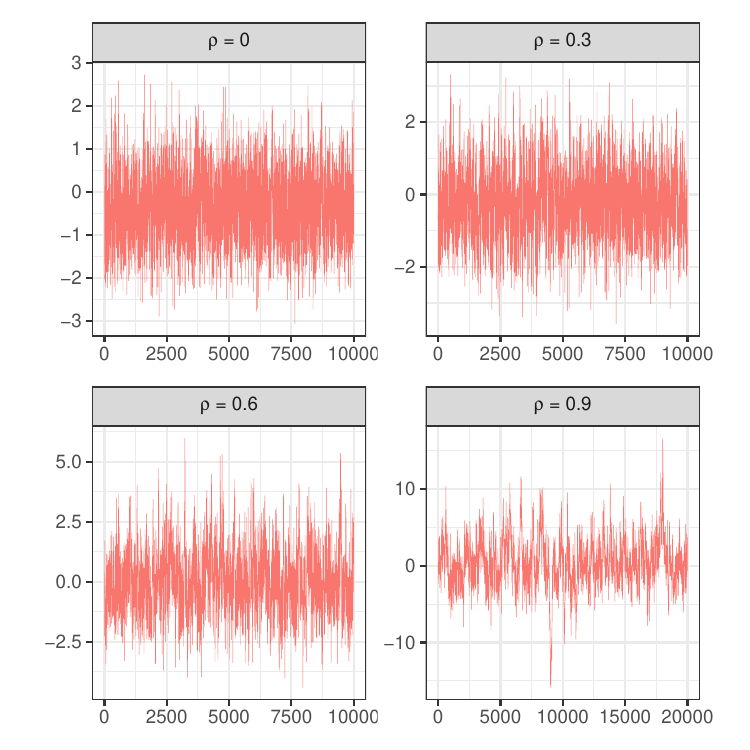} 
    \end{minipage} \hfill 
    \begin{minipage}{0.27\textwidth}
        \centering
        \includegraphics[width=\textwidth]{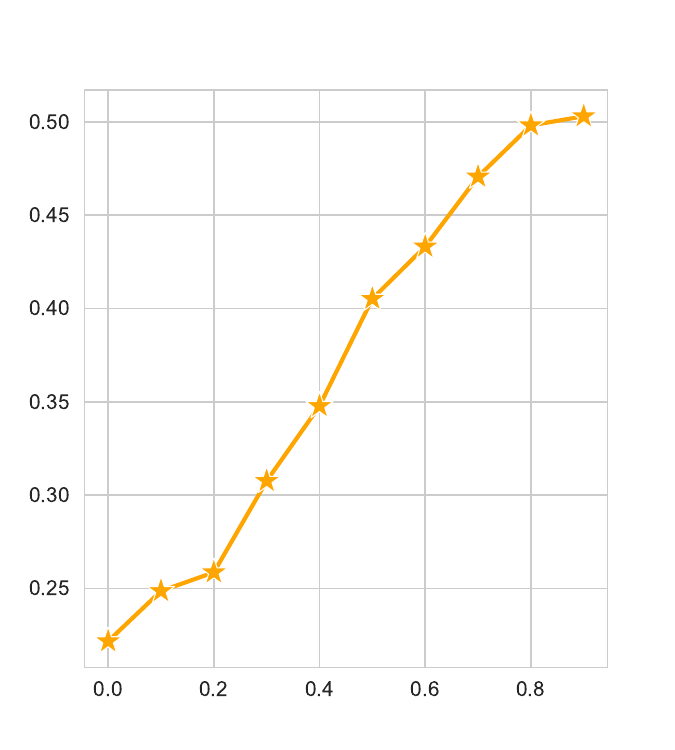} 
    \end{minipage} \hfill 
    \begin{minipage}{0.4\textwidth}
        \centering
        \includegraphics[width=\textwidth]{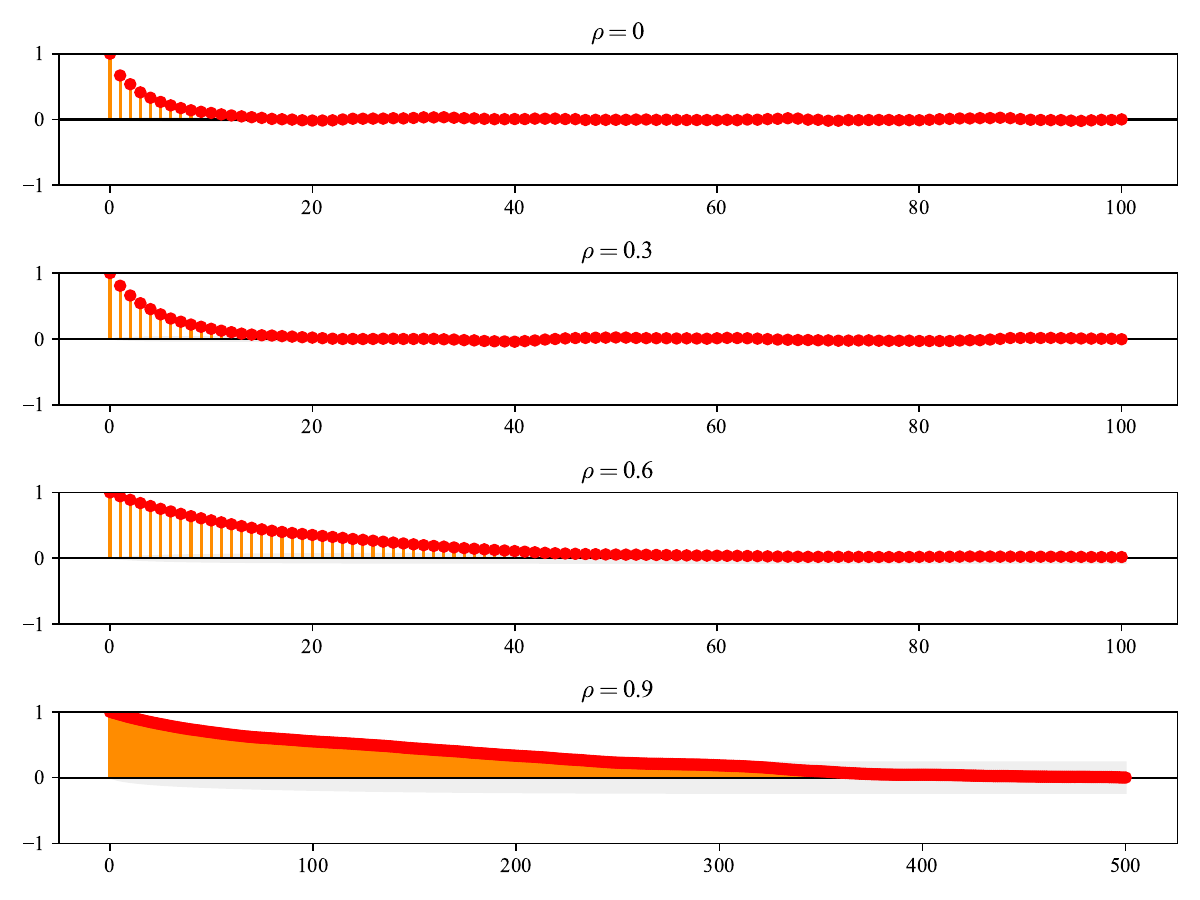}        
    \end{minipage}
    \caption{Diagnostic plots for MALA under setting II. Left panel: trace plots for a randomly selected coordinate in a single realization, for $\rho \in \{0, 0.3, 0.6, 0.9\}$. Middle panel: MALA acceptance rate for $\rho$ between $0$ and $0.9$. Right panel: autocorrelation plots for $\rho \in \{0, 0.3, 0.6, 0.9\}$. }
    \label{fig:setting2-MALA-diagnostics}
\end{figure}

\subsection{Setting II, HMC}
\label{sec:diagnostics4}

Finally, we tune the parameters for HMC algorithm under setting II. We take $\bOmega = \id_d$, $\eps = 0.5$, and $\ell = 10$. We set $B = 10^4$ for $\rho \in \{0, 0.3, 0.6\}$ and $B = 2 \times 10^4$ for $\rho = 0.9$. The diagnostic plots are presented in Figure \ref{fig:setting2-HMC-diagnostics}. 

\begin{figure}[ht]
    \centering
     \begin{minipage}{0.31\textwidth}
        \centering
        \includegraphics[width=\textwidth]{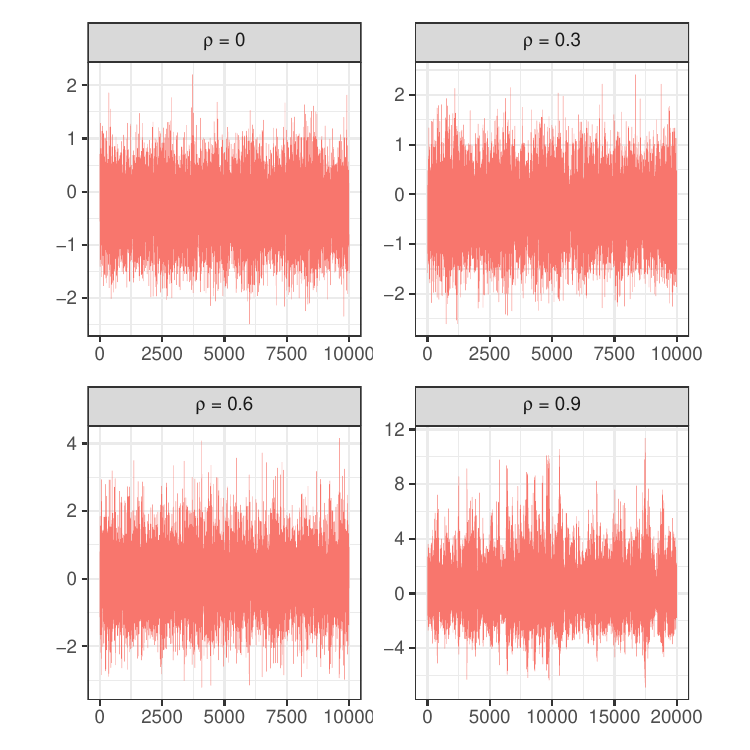} 
    \end{minipage} \hfill 
    \begin{minipage}{0.27\textwidth}
        \centering
        \includegraphics[width=\textwidth]{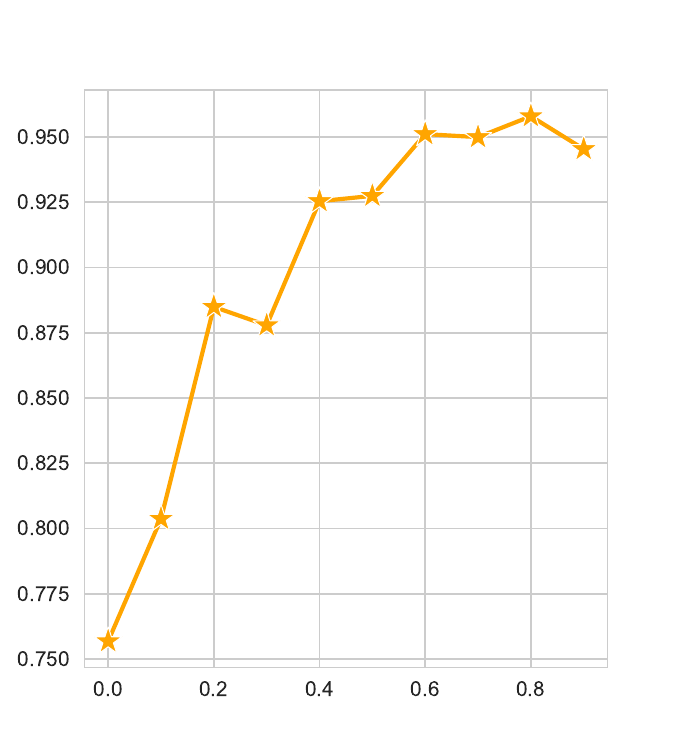} 
    \end{minipage} \hfill 
    \begin{minipage}{0.4\textwidth}
        \centering
        \includegraphics[width=\textwidth]{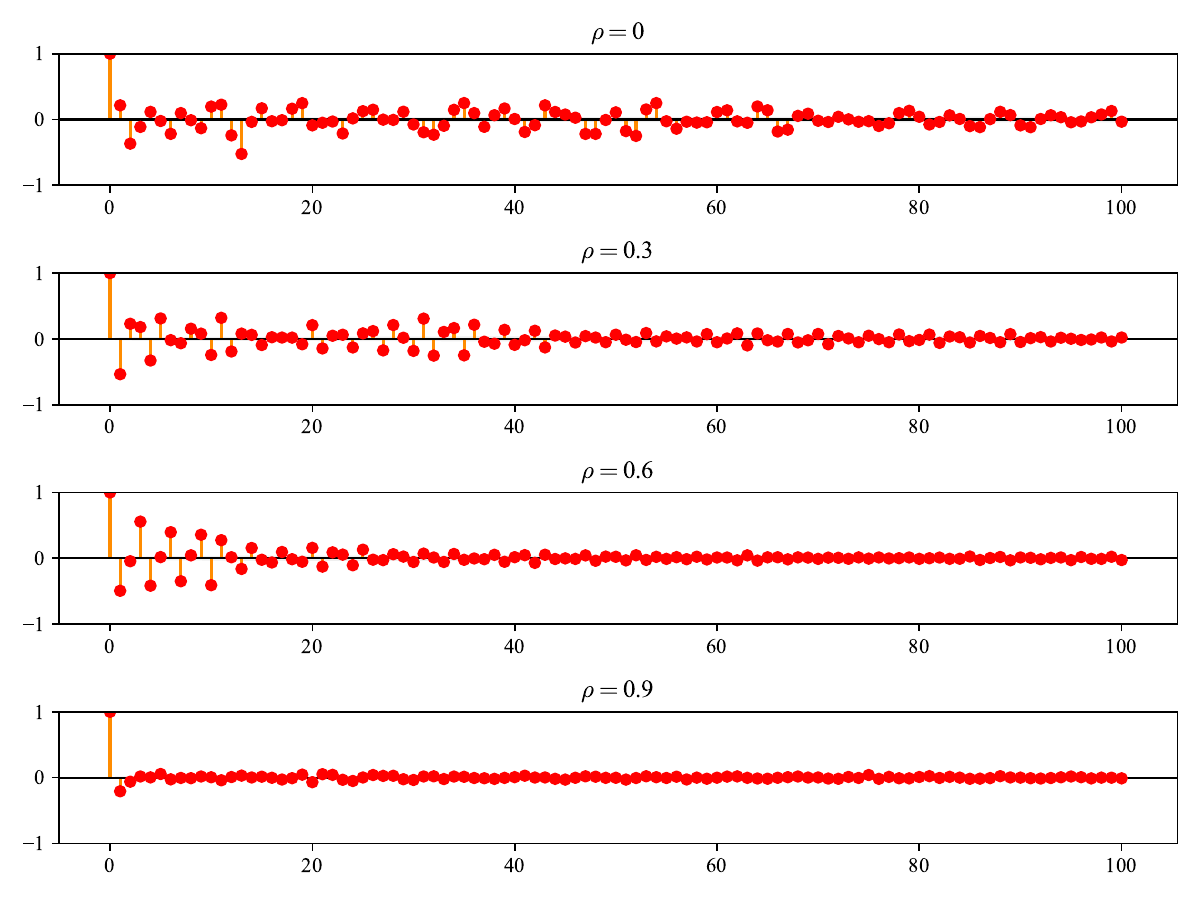}        
    \end{minipage}
    \caption{Diagnostic plots for HMC under setting II. Left panel: trace plots for a randomly selected coordinate in a single realization, for $\rho \in \{0, 0.3, 0.6, 0.9\}$. Middle panel: HMC acceptance rate for $\rho$ between $0$ and $0.9$. Right panel: autocorrelation plots for $\rho \in \{0, 0.3, 0.6, 0.9\}$. }
    \label{fig:setting2-HMC-diagnostics}
\end{figure}

\section{Proofs for random designs}
\label{sec:proof-random-design}

\subsection{Proof of \cref{thm:feasible}}
\label{sec:proof-thm:feasible}

We prove \cref{thm:feasible} in this section. 
To this end, we apply the matrix deviation inequality from \cite[Section 9.1]{vershynin2018high}. 
We copy this inequality below for readers' convenience.
\begin{lemma}[Matrix deviation inequality]
\label{lemma:matrix-deviation-inequality}
	Let $\bX$ be an $n \times d$ matrix whose rows $\bx_i$ are independent, isotropic, and sub-Gaussian random vectors in $\RR^d$. 
	We also assume that $K = \max_i \|\bx_i\|_{\psi_2}$.
	Then, for any subset $T \subseteq \RR^d$ and any $u \geq 0$, the event
	\begin{align*}
		\sup_{\ba \in T} \Big| \|\bX \ba\|_2 - \sqrt{n}\|\ba\|_2 \Big| \leq c K^2 (w(T) + u\, \mathrm{rad}(T))
	\end{align*} 
	holds with probability at least $1 - 2 \exp(-u^2)$. Here, $c$ is a positive numerical constant, and
	\begin{align*}
		\mathrm{rad}(T) = \sup_{\ba \in T} \|\ba\|_2, \qquad w(T) = \EE_{\bg \sim \normal(\mathbf{0}_d, \id_d)}\big[\sup_{\ba \in T} \langle \bg, \ba \rangle\big]. 
	\end{align*}
\end{lemma}

\begin{proof}[Proof of \cref{thm:feasible}]

Next, we apply \cref{lemma:matrix-deviation-inequality} to prove \cref{thm:feasible}. 
To this end, we define $T = \{\ba \in \RR^d: \|\ba\|_2 = 1\}$. We then see that $\mathrm{rad}(T) = 1$ and $w(T) = \EE[\|\bg\|_2] \leq d^{1/2}$.  
Setting $u = r\sqrt{d}$ in \cref{lemma:matrix-deviation-inequality} for some $r \geq 0$, we obtain that with probability at least $1 - 2 \exp({-dr^2})$, 
\begin{align}
\label{eq:Xan}
	\sup_{\|\ba\|_2 = 1} \Big| \|\bX \ba \|_2 - \sqrt{n} \Big| \leq cK^2 \sqrt{d} \, (r + 1). 
\end{align}
Recall that by assumption $n / d \geq C_1 K^4 (r + 1)^2 $. 
As a consequence of that assumption and \cref{eq:Xan}, we conclude that for a large enough $C_1$, 
\begin{align}
\label{eq:lambdamaxmin}
\begin{split}
	& \lambda_{\max}(\sigma_d^{-2} \bX^{\top} \bX) \in \Big[ \sigma_d^{-2} (\sqrt{n} - c K^2 \sqrt{d} \, (r + 1))^2, \,  \sigma_d^{-2} (\sqrt{n} + c K^2 \sqrt{d} \, (r + 1))^2 \Big], \\
	& \lambda_{\min}(\sigma_d^{-2} \bX^{\top} \bX) \in \Big[ \sigma_d^{-2} (\sqrt{n} - c K^2 \sqrt{d} \, (r + 1))^2, \,  \sigma_d^{-2} (\sqrt{n} + c K^2 \sqrt{d} \, (r + 1))^2 \Big]. 
\end{split}
\end{align}
%
Furthermore, via choosing a large enough $C_1$, \cref{eq:lambdamaxmin} implies the following: 
\begin{align}
\label{eq:24}
	\lambda_{\max}(\sigma_d^{-2} \bX^{\top} \bX), \lambda_{\min}(\sigma_d^{-2} \bX^{\top} \bX) \in \big[\sigma_d^{-2} n / 2, \, 2 \sigma_d^{-2} n\big]. 
\end{align}
Taking $\gamma = \lambda_{\max}(\sigma_d^{-2} \bX^{\top} \bX) + K^2 \sigma_d^{-2}$, we obtain that  
\begin{align}
\label{eq:25}
\begin{split}
	\frac{1}{\gamma - \lambda_{\min} (\sigma_d^{-2} \bX^{\top} \bX)} \geq & \frac{\sigma_d^2}{(\sqrt{n} + cK^2 \sqrt{d} \, (r + 1))^2 - (\sqrt{n} - cK^2 \sqrt{d} \, (r + 1))^2 + K^2} \\
	= & \frac{\sigma_d^2}{2cK^2 \sqrt{nd} (r + 1) + K^2} \geq  \frac{c_3 d}{ K^2 \sqrt{nd} (r + 1)},
\end{split}   
\end{align}
where $c_3$ is a positive numerical constant. 
To obtain the second lower bound above, we use the assumption that $c_1 > \sigma_d^2 / d > c_2$ for positive numerical constants $c_1$ and $c_2$.
 
By \cref{lemma:V-gamma-general}, we know that 
\begin{align*}
	\inf_{x \in \RR} V_{\gamma}''(x) \geq - C_0 (\gamma^{-1} + \gamma^{-2}) \cdot (1 + \log (\gamma + 1))^{\frac{2k - 1}{k}},  
\end{align*}
where we recall that $k \in \NN_+$ is a function of $\mu$, and $C_0 > 0$ is a constant that depends only on $(q, \mu)$. 
Using \cref{eq:24} and the assumption that $c_1 > \sigma_d^2 / d > c_2$, we conclude that 
\begin{align*}
	\inf_{x \in \RR} V_{\gamma}''(x) \geq - C_0 \bar c_k (d / n + d^2 / n^2) \cdot (1 + \log (n / d + 1))^{\frac{2k - 1}{k}}
\end{align*}
for some positive constant $\bar c_k$ that depends only on $k$. 
Putting together the above lower bound and \cref{eq:25}, a sufficient condition for \cref{eq:def-feasible} to hold is 
\begin{align*}
	\frac{c_3 \sqrt{d}}{ K^2 \sqrt{n} (r + 1)} > C_0 \bar c_k (d / n + d^2 / n^2) \cdot (1 + \log (n / d + 1))^{\frac{2k - 1}{k}}. 
\end{align*}
The proof is complete by setting $r = 1$. 
\end{proof}

\subsection{Proof of  \cref{thm:asymp-feasible}}
\label{sec:proof-thm:asymp-feasible}

\begin{proof}[Proof of \cref{thm:asymp-feasible}]
	If \cref{eq:asymp-feasible} holds, then we choose $\gamma$ that satisfies both inequalities in \cref{eq:asymp-feasible}. 
	Invoking Bai-Yin's law, we conclude that 
	\begin{align}
	\label{eq:26}
		\lambda_{\max}(\sigma_d^{-2} \bX^{\top} \bX) \overset{a.s.}{\to} \sigma_0^{-2} \delta (1 + 1 / \sqrt{\delta})^2, \qquad \lambda_{\min}(\sigma_d^{-2} \bX^{\top} \bX) \overset{a.s.}{\to} \sigma_0^{-2} \delta (1 - 1 / \sqrt{\delta})^2\mathbbm{1}\{\delta \geq 1\}. 
	\end{align}
	Therefore, with probability $1 - o_n(1)$ \cref{eq:def-feasible} holds. In this case the problem is feasible.
	
	On the other hand, if for all $\gamma \geq {\delta(1 + 1 / \sqrt{\delta})^2}{ \sigma_0^{-2}}$, it holds that 
	\begin{align*}
		\frac{1}{\gamma - \delta \sigma_0^{-2}(1 - 1 / \sqrt{\delta})^2 \mathbbm{1}\{\delta \geq 1\}} <- \inf_{x \in \RR} V_{\gamma}''(x).  
	\end{align*}
	By \cref{eq:26}, we know that for all $\gamma > \lambda_{\max}( \sigma_d^{-2} \bX^{\top} \bX)$, it must be the case that $\gamma > \sigma_0^{-2} \delta (1 + 1 / \sqrt{\delta})^2 + o_P(1)$. 
	For such $\gamma$,  
	\begin{align*}
		\frac{1}{\gamma - \lambda_{\min}(\sigma_d^{-2} \bX^{\top} \bX)} = &\frac{1}{\gamma - \delta \sigma_0^{-2}(1 - 1 / \sqrt{\delta})^2 \mathbbm{1}\{\delta \geq 1\}} + o_P(1). 
	\end{align*}
	By continuity, 
	the above equation is with probability $1 - o_n(1)$ strictly smaller than $-\inf_{x \in \RR} V_{\gamma}''(x)$. Therefore, the problem is with probability $1 - o_n(1)$ not feasible.

\end{proof}

\section{Additional simulation details}

We present the pseudo code for HMC in this section  

\begin{algorithm}
\caption{Hamiltonian Monte Carlo (HMC)}\label{alg:HMC}
\begin{algorithmic}[1]
\REQUIRE mass matrix $\bOmega$, leapfrog stepsize $\epsilon$, number of leapfrog steps $\ell$, Monte Carlo steps $K$;
\STATE Get an estimate of $\bvarphi^{\ast}$ via gradient ascent, and denote it by $\widehat\bvarphi^{\ast}$; 
\STATE Initialize HMC at $\bvarphi_0 \sim \normal(\hat\bvarphi^{\ast}, L^{-1} \id_d)$, where $L = 10$;
\FOR{$k = 1, 2, \cdots, K$}
	\STATE $\mathtt{Proceed} \gets \mathtt{False}$;
	\WHILE{$\mathtt{Proceed} = \mathtt{False}$}
	\STATE $\brho \sim \normal(\mathbf{0}, \bOmega)$, $\bvarphi \gets \bvarphi_{k - 1}$;
	\FOR{$i = 1, 2, \cdots, \ell$}
		\STATE $\brho \gets \brho - \epsilon \cdot \nabla H (\bvarphi) / 2$; 
		\STATE $\bvarphi \gets \bvarphi + \epsilon \cdot   \bOmega^{-1} \brho $; 
		\STATE $\brho \gets \brho - \epsilon \cdot \nabla H (\bvarphi) / 2$;
	\ENDFOR
	\STATE $\alpha \gets \min \{0, -H(\bvarphi) + H(\bvarphi_{k - 1}) - \bvarphi^{\top} \bOmega^{-1} \bvarphi / 2 + \bvarphi_{k - 1}^{\top} \bOmega^{-1} \bvarphi_{k - 1} / 2 \}$;
	\STATE Sample $u \sim \Unif[0,1]$;
	\IF{$u \leq e^\alpha$}
		\STATE $\mathtt{Proceed} \gets \mathtt{True}$; 
		\STATE $\bvarphi_k \gets \bvarphi$;
	\ENDIF
	\ENDWHILE
\ENDFOR
\end{algorithmic}
{\bf Return}: $\{\bvarphi_k: k \in [K]\}$.
\end{algorithm}

\end{appendices}

\end{document}